\documentclass{article}
\usepackage[utf8]{inputenc}
\pdfoutput=1

\usepackage[left=25mm, right=25mm, top=20mm, bottom=20mm]{geometry}
\usepackage{amsfonts}
\usepackage{amsmath}
\usepackage{amssymb}
\usepackage{amsthm}
\usepackage{authblk}
\usepackage{bm,bbm}
\usepackage{booktabs}
\usepackage{cancel}
\usepackage{caption}
\usepackage{color}
\usepackage[mathscr]{euscript}
\usepackage{hyperref}
\usepackage{physics}
\usepackage{simpler-wick}
\usepackage{subcaption}
\usepackage{tikz}
\usetikzlibrary{cd}
\usetikzlibrary{calc,decorations.markings}  
\usetikzlibrary{arrows.meta}

\hypersetup{colorlinks,breaklinks,
            linkcolor = purple,
            citecolor = teal,
            urlcolor = purple,
            }
\definecolor{darkblue}{rgb}{0.1,0.1,0.7}
\definecolor{darkred}{rgb}{0.5,0.1,0.1}
\definecolor{darkgreen}{rgb}{0.0,0.42,0.06}

\newlength{\dhatheight}

\tikzset{->-/.style 2 args={
    postaction={decorate},
    decoration={markings, mark=at position #1 with {\arrow[#2]{latex}}}
    },
    ->-/.default={0.5}{}
}
\tikzset{-<-/.style 2 args={
    postaction={decorate},
    decoration={markings, mark=at position #1 with {\arrow[#2]{latex reversed}}}
    },
    -<-/.default={0.5}{}
}

\newcommand{\rmT}{{\scalebox{0.5}{$\mathrm T$}}}

\theoremstyle{definition}
\newtheorem{theorem}{Theorem}[section]
\newtheorem{proposition}[theorem]{Proposition}

\newtheorem{definition}[theorem]{Definition}
\newtheorem{conjecture}[theorem]{Conjecture}
\newtheorem{lemma}[theorem]{Lemma}

\theoremstyle{remark}
\newtheorem{example}[theorem]{Example}
\newtheorem{remark}[theorem]{Remark}

\newcommand{\ii}{\mathrm{i}}

\numberwithin{equation}{section}

\title{Gauge origami and quiver W-algebras II:\\ Vertex function and beyond quantum $q$-Langlands correspondence}
\author[$\diamondsuit$]{Taro Kimura}
\author[$\spadesuit$]{Go Noshita}
\affil[$\diamondsuit$]{Institut de Mathématiques de Bourgogne, Université Bourgogne Europe, CNRS, France\footnote{Unité Mixte de Recherche (UMR 5584) commune au CNRS et à l’Université Bourgogne Europe}}
\affil[$\spadesuit$]{Department of Physics, The University of Tokyo, Tokyo, Japan}

\date{}

\begin{document}

\maketitle

\begin{abstract}
    We continue the study of generalized gauge theory called gauge origami, based on the quantum algebraic approach initiated in \cite{Kimura:2023bxy}.
    In this article, we explore the D2 brane system realized by the screened vertex operators of the corresponding W-algebra. 
    The partition function of this system given by the corresponding conformal block is identified with the vertex function associated with quasimaps to Nakajima quiver varieties and generalizations, that plays a central role in the quantum $q$-Langlands correspondence.
    Based on the quantum algebraic perspective, we address three new aspects of the correspondence:
    (i) Direct equivalence between the electric and magnetic blocks by constructing stable envelopes from the chamber structure of the vertex operators, (ii) Double affine generalization of quantum $q$-Langlands correspondence, and (iii) Conformal block realization of the origami vertex function associated with intersection of quasimaps, that realizes the higher-rank multi-leg Pandharipande--Thomas vertices of 3-fold and 4-fold.
\end{abstract}

\setcounter{tocdepth}{2}
\tableofcontents

\section{Introduction and summary}

Gauge origami, introduced and investigated in a series of works by Nekrasov \cite{Nekrasov:2015wsu,Nekrasov:2016qym,Nekrasov:2016ydq,Nekrasov:2017rqy,Nekrasov:2017gzb}, deserves a generalization of gauge theory, formulated as a union and intersection of gauge theories supported on various subspaces in $\mathbb{C}^4$.
It is formulated as intersecting coherent sheaves on the subschemes, and the equivariant integral of the corresponding moduli space gives rise to the gauge origami partition function, i.e., the generating function of the corresponding enumerative invariants.
In this article, we continue the study of such a generalized gauge theory based on the quantum algebraic structure initiated in our previous work~\cite{Kimura:2023bxy}.%
\footnote{See \S\ref{sec:VO} for the summary of the vertex operator formalism used in this approach. See also \cite[\S3]{Kimura:2023bxy} for the gauge origami systems and their partition functions.}
We in particular focus on the D2 brane system algebraically realized by the screening current of the corresponding W-algebra.
This system is known to be related to quasimaps $f$ from $\mathbb{P}^1$, a compactification of $\mathbb{C}$, to the corresponding target space $X$, typically taken to be a Nakajima quiver variety.
The corresponding partition function, a.k.a., the vertex function, is given by the generating function of equivariant (K-theoretic) Euler characteristics of the moduli space of quasimaps $\mathsf{QM}(X) = \{f : \mathbb{P}^1 \to X \}/\sim$~\cite{Okounkov:2015spn}, that is identified with the vortex partition function in the physics context.
See, e.g.,~\cite{Dedushenko:2023qjq,Crew:2023tky}.
The vertex function constructed geometrically also has an algebraic realization addressed by Aganagic, E.~Frenkel, and Okounkov in the context of the \emph{quantum $q$-Langlands correspondence}.
\begin{theorem}[\cite{Aganagic:2017smx}]\label{thm:AFO}
    The vertex function is the deformed conformal block of the $q$-deformed W-algebra.
\end{theorem}
The quantum $q$-Langlands correspondence is the two-parameter deformation of the original geometric Langlands correspondence, that is the equivalence between the spaces of deformed conformal blocks of the quantum affine algebra $U_q(\widehat{\mathfrak{g}})$ (electric block) and those of the $q$-deformed W-algebra W$_{q_{1,2}}({}^L\mathfrak{g})$ (magnetic block), where $^L\mathfrak{g}$ is the Langlands dual of a simple Lie algebra $\mathfrak{g}$.
From the analogy with the ordinary Langlands correspondence, the deformed conformal block is also referred to as the conformal block.
Theorem~\ref{thm:AFO} plays a central role in the correspondence since the vertex function is a building block of a solution to the quantum Kniznik--Zamolodchikov ($q$-KZ) equation introduced by I.~Frenkel and Reshetikhin~\cite{Frenkel:1991gx} for the conformal block of the quantum affine algebra.
One can construct the integral formula of the solution, a.k.a., the $q$-hypergeometric integral solution, to the $q$-KZ equation from the vertex function together with the stable envelope~\cite{Maulik:2012wi}. 
In the context of the $q$-hypergeometric solution, the stable envelope is identified with the so-called weight function~\cite{Gorbounov:2012uu,Rimanyi:2014JPG,Felder:2017SIGMA,Konno:2017mos}.
See, e.g., \cite{Jeong:2023qdr,Tamagni:2023wan,Haouzi:2023doo,Jeong:2024hwf} for recent gauge theory approaches to the Langlands correspondence.

In this article, we address three new perspectives concerning the quantum $q$-Langlands correspondence as follows.
\subsection{Direct equivalence between electric and magnetic blocks}
The $q$-KZ equation is a $q$-difference equation for the electric conformal block of the quantum affine algebra $U_q(\widehat{\mathfrak{g}})$ that takes a value in the tensor product module of the corresponding representations.
Meanwhile, the magnetic conformal block of the $q$-deformed W-algebra W$_{q_{1,2}}(^L\mathfrak{g})$ is apparently insensitive to the representation of $U_q(\widehat{\mathfrak{g}})$ by definition, so that we need to insert the stable envelope by hand to construct the $q$-KZ solution from the magnetic block side.
In this article, we reexamine the W-algebra conformal block, and we figure out how to impose the representation data of the conformal block, i.e., the stable envelope purely from the W-algebra side. 
The key idea is the radial ordering and the chamber structure of the vertex operators.
This implies that we do not need to insert the stable envelope by hand to the W-algebra block to construct the $q$-hypergeometric solution to the $q$-KZ equation.
In fact, the W-algebra block directly solves the $q$-KZ equation.
\begin{theorem}[Theorem~\ref{thm:e_qKZ}]
    The elliptic conformal block of the screened vertex operators of the $q$-deformed W-algebra W$_{q_{1,2}}(A_1)$ obeys the dynamical elliptic $q$-KZ equation.
\end{theorem}
We emphasize that this is a reformulation of the result given by Konno~\cite{Konno:2017mos,Konno:2018JIS,Konno2020}, who proved the corresponding statement in the context of elliptic quantum group $U_{q,p}(\widehat{\mathfrak{sl}}_n)$, where the so-called $e$ and $f$ currents are known to be equivalent to the screening currents of the W-algebra W$_{q_{1,2}}(A_n)$%
\footnote{%
We also use the Dynkin classification notation, e.g, $\mathrm{W}_{q_{1,2}}(A_{n-1}) = \mathrm{W}_{q_{1,2}}(\mathfrak{sl}_n)$. 
} for level-one representations.
As pointed out by Konno, the quantum $q$-Langlands correspondence becomes more direct in the elliptic setup:
We show that the elliptic stable envelope introduced by Aganagic and Okounkov~\cite{Aganagic:2016jmx} naturally arises from the chamber structure of the vertex operators.
However, it has been known that the elliptic envelope is not sufficient to construct the $q$-KZ solution in the trigonometric setup.
We need to also insert the trigonometric version of the stable envelope to obtain the trigonometric $q$-KZ solution.
In fact, we do not need the additional stable envelope in the elliptic case, where the elliptic envelope is sufficient to obtain the elliptic $q$-KZ solution.
We prove this Theorem for the case $\mathfrak{g} = \mathfrak{sl}_2$ associated with $A_1$ quiver.
Its generalization to simply-laced algebras seems straightforward. 
It would be also interesting to consider the non-simply-laced cases~\cite{Frenkel:1997lee,Kimura:2017hez,KPfractional}.

\subsection{Double affine quantum $q$-Langlands correspondence}
The vertex function is constructed for a wide class of Nakajima quiver varieties.
A particularly interesting family of quiver varieties is those associated with cyclic quivers (untwisted affine $A$-type quivers).
It has been known that affine $A$-type quiver varieties play an important role in the context of geometric representation theory of quantum toroidal algebra of $\mathfrak{gl}_n$, which is a double affine quantum algebra. 
See, e.g.,~\cite{Hernandez:2009SM,Tsymbaliuk:2022bqx,Matsuo:2023lky} for reviews.
Meanwhile, the formalism of quiver W-algebra~\cite{Kimura:2015rgi,Kimura:2016dys,Kimura:2017hez} allows us to consider the W-algebra side of this story beyond the construction of E. Frenkel and Reshetikhin for finite-type simple Lie algebras~\cite{Frenkel:1997lee}.
One can construct the $q$-deformed W-algebra associated with cyclic quivers denoted by W$_{q_{1,2,3,4}}(\widehat{A}_{n-1})$ that depends on four parameters $\underline{q} = (q_1,q_2,q_3,q_4)$ under the condition $q_1 q_2 q_3 q_4 = 1$ (Calabi--Yau 4-fold condition).
This affine quiver W-algebra deserves a double affine W-algebra.
\begin{theorem}[Theorems~\ref{thm:vertex_Fock}, \ref{thm:vertex_MacMahon}]
    The vertex function for Nakajima quiver variety associated with the Jordan quiver ($\widehat{A}_0$ quiver) is given by the conformal block of affine quiver W-algebra W$_{q_{1,2,3,4}}(\widehat{A}_{0})$.
\end{theorem}
This is a double affine version of Theorem~\ref{thm:AFO}, establishing the double affine quantum $q$-Langlands correspondence between the quantum toroidal $\mathfrak{gl}_1$ and the affine quiver W-algebra W$_{q_{1,2,3,4}}(\widehat{A}_{0})$.
In fact, it has been recently shown by Smirnov that the quantum difference equation associated with the quantum toroidal $\mathfrak{gl}_1$ is solved by the vertex function of the Hilbert scheme of points on $\mathbb{C}^2$~\cite{Smirnov:2021cyf}.
We prove this Theorem for the rank-one case, whereas a higher-rank generalization is also straightforward.

We remark that there are two types of highest-weight modules of quantum toroidal $\mathfrak{gl}_1$, Fock module and MacMahon module.
Geometrically, the former one is constructed from the quiver variety of the Jordan quiver, namely the Hilbert scheme of points on $\mathbb{C}^2$ for the rank-one case, and the corresponding Quot scheme for the higher-rank case.
On the other hand, the latter one has a geometric realization by a quiver with a potential, which gives rise to the Hilbert scheme of points and the Quot scheme on $\mathbb{C}^3$.
In fact, the conformal block of W$_{q_{1,2,3,4}}(\widehat{A}_{0})$ reproduces vertex functions associated with both types of modules depending on ``the momentum'' of the vertex operator.
This is analogous to the finite-type case, where the momentum parametrizes the highest-weight module of the quantum affine algebra.

\subsection{Origami vertex function, PT vertex, and W-algebra}

There is a correspondence between the moduli space of quasimaps to the Hilbert scheme of points on $\mathbb{C}^2$ and that of the stable pairs that gives rise to the Pandharipande--Thomas (PT) invariant, a cousin of the Donaldson--Thomas (DT) invariant, of $\mathbb{C}^3$ with a boundary condition in one specific direction.
This correspondence has been recently generalized by Cao and Zhao~\cite{Cao:2023lon} to the higher-dimensional situation $\mathbb{C}^4$ where the vertex function associated with the Hilbert scheme of points on $\mathbb{C}^3$ plays a role of the generating function of the PT invariants for $\mathbb{C}^4$, i.e., the one-leg PT4 vertex.
The PT vertex~\cite{Pandharipande:2007sq} is the generating function of the PT invariants following the DT vertex of Maulik--Nekrasov--Okounkov--Pandharipande (MNOP)~\cite{Maulik:2003rzb} motivated by the topological vertex formalism introduced in the context of the topological string theory~\cite{Aganagic:2003db}.
The PT vertex is also introduced for $\mathbb{C}^4$ and the so-called DT/PT correspondence for this case has been examined in detail~\cite{Cao:2019tnw,Cao:2019tvv,Liu:2024bgp}. 
The higher-rank version of these enumerative invariants and their correspondence have been studied by Toda~\cite{Toda:2020AG}.
In general, one can assign boundary conditions in three (four) directions for $\mathbb{C}^3$ ($\mathbb{C}^4$) to define the so-called multi-leg PT vertex, for which we propose an algebraic construction.
\begin{conjecture}[Conjecture~\ref{conj:origami_vertex}]
    The conformal block of multi-screened vertex operators of the W-algebra W$_{q_{1,2,3,4}}(\widehat{A}_{0})$ is the (higher-rank) multi-leg PT vertex of $\mathbb{C}^3$ and $\mathbb{C}^4$.
\end{conjecture}
We call this generalized version of the vertex function associated with the multi-leg situation the \emph{origami vertex function} as is naturally given by the gauge origami partition function.
We derive the contour integral formula of the conformal block and check that its pole structure, corresponding to the equivariant fixed points in the moduli space, is consistent with the PT count in $\mathbb{C}^3$ and $\mathbb{C}^4$ for various boundary conditions.
Our formula for the one-point conformal block reproduces the known result for the rank-one PT3 and PT4 vertices, and we conjecture that the $n$-point block generates the rank-$n$ PT3 and PT4 invariants.

The original PT vertex is based on the so-called curve counting, hence it corresponds to the D2 brane system as mentioned above. 
Recently, its higher dimensional analog, the surface counting, has been discussed~\cite{Bae:2022pif,Bae:2024bpx}.
See also a related work~\cite{Nekrasov:2023nai}.
We expect that the surface counting is similarly formulated based on the vertex operators by replacing the screening current with the corresponding D4 brane vertex operators~\cite{Kimura:2023bxy}.
A further analysis is left for a future study.


\subsection*{Acknowledgements}

We would like to thank Y.~Cao, G.~Felder, H.~Konno, N.~Lee, M.~Kool, and S.~Monavari for valuable communications.
The work of TK was supported by CNRS through MITI interdisciplinary programs, EIPHI Graduate School (No.~ANR-17-EURE-0002) and Bourgogne-Franche-Comté region. GN is supported by JSPS KAKENHI Grant-in-Aid for JSPS fellows Grant No.~JP22J20944, JSR Fellowship, and FoPM (WINGS Program), the University of Tokyo.

\subsection*{Notations}

\begin{itemize}
    \item 
    We write $x_i/x_j = x_{ij}$, $a_i/a_j = a_{ij}$, etc, for $\mathbb{C}^\times$-variables.

    \item 
    We define
    $\underline{\mathbf{2}} = \{1,2\}$,
    $\underline{\mathbf{3}} = \{1,2,3\}$,
    $\underline{\mathbf{4}} = \{1,2,3,4\}$.
    For $A \subseteq \underline{\mathbf{4}}$, we write $q_A = \prod_{\sigma \in A} q_\sigma$.

    \item 
    We write
    \begin{align}
        f(q_{\sigma,\sigma'} x) = f(q_\sigma x) f(q_{\sigma'} x)
        \, , \qquad 
        f(x^{\pm 1} ) = f(x^{+1}) f(x^{-1})
        \, .
    \end{align}

    \item 
    We define the $q$-shifted factorial for $q \in \mathbb{C}^\times$,
    \begin{align}
        (x;q)_d = \prod_{n=0}^{d-1} (1 - x q^n)
        \, .
    \end{align}
    For $d \in \mathbb{Z}$ in general, we define
    \begin{align}
        (x;q)_d = \frac{(x;q)_\infty}{(q^d x;q)_\infty} =
        \begin{cases}
            (1-x)(1-qx) \cdots (1-q^{d-1}x) & (d > 0) \\
            (1-q^d x)^{-1} \cdots (1-q^{-1}x)^{-1} & (d < 0)
        \end{cases}
    \end{align}
    We write
    \begin{align}
        (a_1,\ldots,a_r;q)_d = \prod_{i=1}^r (a_i;q)_d \, .
    \end{align}  

    \item 
    We define the theta function for $p \in \mathbb{C}^\times$ with $|p| < 1$,
    \begin{align}
        \theta(z;p) = (z,pz^{-1};p)_\infty = (z;p)_\infty (pz^{-1};p)_\infty \, . \label{eq:theta_fn}
    \end{align}   
    We write
    \begin{align}
        \theta(a_1,\ldots,a_r;q) = \prod_{i=1}^r \theta(a_i;q) \, .
    \end{align} 
\end{itemize}

\section{W$_{q_{1,2}}(A_1)$ conformal block}

We consider the conformal block of the $q$-deformed W-algebra associated with $A_1$ quiver denoted by W$_{q_{1,2}}(A_1)$.
This algebra is identified with the $q$-deformed Virasoro algebra introduced by Shiraishi, Kubo, Awata, and Odake~\cite{Shiraishi:1995rp}.

We denote the OPE factor associated with the vertex operators $\mathsf{V}$ and $\mathsf{W}$ by
\begin{align}
    \wick{ \c{\mathsf{V}}(x) \c{\mathsf{W}}(x') } = \frac{\mathsf{V}(x) \mathsf{W}(x')}{: \mathsf{V}(x) \mathsf{W}(x') :} \, ,
\end{align}
where we denote the normal ordering symbol by $:\cdot:$, moving the annihilation operators to the right and the creation operators to the left.
In this notation, we assume the radial ordering $x \succ x'$, which means $|x| > |x'|$.

\begin{definition}\label{def:OPEs}
    Let $\sigma \in \underline{\textbf{2}}$, $\bar\sigma = \underline{\textbf{2}} \backslash \sigma$, $|q_{1,2}| > 1$, and $x, x', z, K \in \mathbb{C}^\times$.
    We define the vertex operators $\{\mathsf{S}_{1,2}, \mathsf{Z}_K\}$ associated with $A_1$ quiver, obeying the following OPEs,
    \begin{subequations}\label{eq:OPE}
    \begin{align}
        \wick{ \c {\mathsf{S}}_\sigma(x) \c {\mathsf{S}}_\sigma(x') } = \frac{(x'/x;q_\sigma^{-1})_\infty (q_{12}^{-1} x'/x;q_\sigma^{-1})_\infty}{(q_\sigma^{-1} x'/x;q_\sigma^{-1})_\infty (q_{\bar{\sigma}} x'/x;q_\sigma^{-1})_\infty}
        \, , \qquad &
        \wick{ \c {\mathsf{S}}_\sigma(x) \c {\mathsf{S}}_{\bar{\sigma}}(x') } = \frac{(x')^2 q_\sigma}{(1 - q_\sigma x'/x)(1 - q_{\bar{\sigma}}^{-1} x'/x)}
        \, , \\
        \wick{ \c {\mathsf{S}}_\sigma(x) \c {\mathsf{Z}}_K(x') } = \frac{(K x'/x;q_\sigma^{-1})_\infty}{(x'/x;q_\sigma^{-1})_\infty}
        \, , \qquad &
        \wick{ \c {\mathsf{Z}}_K(x') \c {\mathsf{S}}_\sigma(x) } = \frac{(q_\sigma^{-1} x/x';q_\sigma^{-1})_\infty}{(K^{-1} q_\sigma^{-1}x/x';q_\sigma^{-1})_\infty} 
        \, , \\
        \mathsf{S}_\sigma(x) z = (q_{\bar{\sigma}}^{-2} z) \mathsf{S}_\sigma(x)
        \, ,  \qquad &
        \mathsf{Z}_K(x) z = (K z) \mathsf{Z}_K(x)
        \, .
        \label{eq:OPE_z}
    \end{align}
    \end{subequations}
\end{definition}
See, e.g.,~\cite{Kimura:2015rgi} for explicit realization of these vertex operators in terms of the free fields.
The operator $\mathsf{Z}_K$ is written in terms of the vertex operator $\mathsf{V}$ of~\cite[eq.~(3.58)]{Kimura:2015rgi} as $\mathsf{Z}_K(x) = \mathsf{V}(x)/\mathsf{V}(K^{-1}x)$.
We slightly change the zero mode notation compared to \cite{Kimura:2015rgi} to simplify the discussion below.
In this article, all the computations are based on the OPEs defined above.

In the notation of \cite[Proposition C.3]{Kimura:2023bxy}, the vertex operators $\mathsf{S}_{1,2}$ are called the screening currents associated with the plane $\mathbb{C}_{1,2}$, that commute with each other, $\wick{ \c {\mathsf{S}}_1(x) \c {\mathsf{S}}_2(x') } = \wick{ \c {\mathsf{S}}_2(x') \c {\mathsf{S}}_1(x) }$.
As discussed in \cite{Kimura:2023bxy}, the screening current plays a role of the creation operator of D2 branes supported on the plane, and the partition function of this system is realized by the corresponding conformal block.
Meanwhile, the vertex operator $\mathsf{Z}_K$ is associated with D8 ($\overline{\text{D8}}$) branes supported on the whole space $\mathbb{C}^4$.
The screening currents commute with the generating current of the deformed Virasoro algebra ($q$-Virasoro algebra) up to the total difference term.
We denote this algebra by W$_{q_{1,2}}(A_1)$ since it deserves the $q$-deformed W-algebra associated with $A_1$ quiver~\cite{Kimura:2015rgi}.
As seen below, it has an essential connection to quantum affine algebra $U_q(\widehat{\mathfrak{sl}}_2)$ and also elliptic quantum group $U_{q,p}(\widehat{\mathfrak{sl}}_2)$.

\begin{lemma}\label{lemma:OPE_flip}
    We may rewrite the OPEs \eqref{eq:OPE} for $\sigma \in \underline{\mathbf{2}}$ as follows,
    \begin{subequations}
    \begin{align}
        \wick{ \c {\mathsf{S}}_\sigma(x) \c {\mathsf{S}}_\sigma(x') } & = \frac{(x'/x;q_\sigma^{-1})_\infty (x/x';q_\sigma^{-1})_\infty}{(q_{\bar{\sigma}} x'/x;q_\sigma^{-1})_\infty (q_{\bar{\sigma}} x/x';q_\sigma^{-1})_\infty} \frac{\theta(q_{\bar{\sigma}}x'/x;q_\sigma^{-1})}{\theta(x/x';q_\sigma^{-1})} , \\
        \wick{ \c {\mathsf{S}}_\sigma(x) \c {\mathsf{Z}}_K(x') } & = \wick{ \c {\mathsf{Z}}_K(x') \c {\mathsf{S}}_\sigma(x) } \frac{\theta(K x'/x;q_\sigma^{-1})}{\theta(x'/x;q_\sigma^{-1})} \, .
    \end{align}
    \end{subequations}    
\end{lemma}
\begin{proof}
    It directly follows from Definition~\ref{def:OPEs} together with the definition of the theta function~\eqref{eq:theta_fn}.
\end{proof}

\subsection{Single screening current}

In this subsection, we set $\mathsf{S}(x) \equiv \mathsf{S}_2(x)$.
For $|q_\sigma|>1$, define a kernel function,
\begin{align}
    \varphi_{\sigma,K}(\xi,\eta) = \varphi_{K}(\xi,\eta;q_\sigma^{-1}) = \frac{\theta(\xi \eta;q_\sigma^{-1}) \theta(K;q_\sigma^{-1})}{\theta(K \xi;q_\sigma^{-1}) \theta(\eta;q_\sigma^{-1})}
    \, ,
\end{align}
which is quasiperiodic with respect to each variable,
\begin{align}
    \varphi_{\sigma,K}(q_\sigma^{-1} \xi, \eta) = (\eta/K) \varphi_{\sigma,K}(\xi, \eta)
    \, , \quad 
    \varphi_{\sigma,K}(\xi, q_\sigma^{-1} \eta) = \xi^{-1} \varphi_{\sigma,K}(\xi, \eta)
    \, , \quad 
    \varphi_{\sigma,q_\sigma^{-1} K}(\xi, \eta) = \xi \varphi_{\sigma,K}(\xi, \eta)
    \, .
\end{align}
For $|q_\sigma| < 1$, we instead use the theta function $\theta(\cdot;q_\sigma)$ to define the kernel function, namely $\varphi_K(\xi,\eta;q_\sigma)$.

\begin{definition}\label{def:Z_screened}
    Let $z \in \mathbb{C}^\times$.
    We define the $m$-screened vertex operator,
    \begin{align}
        \mathsf{Z}_K^{(m)}(a) = \oint_\gamma [\dd{\underline{x}}] \left( \prod_{i=1,\ldots,m}^{\curvearrowleft} \varphi_{2,K}(a/x_i,zK) \mathsf{S}(x_i) \right) \mathsf{Z}_K(a)
        \, , \qquad
        [\dd{\underline{x}}] = \prod_{i=1}^m \frac{\dd{x}_i}{2 \pi \ii x_i} \, ,
    \end{align}
    where we apply the ordered product under the radial ordering configuration $x_m \succ x_{m-1} \succ \cdots \succ x_2 \succ x_1 \succ a$,
    \begin{align}
        \prod_{i=1,\ldots,m}^{\curvearrowleft} \mathsf{V}(x_i) = \mathsf{V}(x_m) \mathsf{V}(x_{m-1}) \cdots \mathsf{V}(x_2) \mathsf{V}(x_1) .
    \end{align}
    The contour $\gamma$ is taken such that the integral is invariant under shift of the variables, $x_i \to q_2^{-1} x_i$.
\end{definition}
This definition is motivated by the construction of the intertwining operator of the elliptic quantum group $U_{q,p}(\widehat{\mathfrak{sl}}_2)$ in \cite[\S5.3]{Jimbo:1998bi}.%
\footnote{%
The definition of the screened vertex operator given in Definition~\ref{def:Z_screened} is specific to the rank-one theory. 
See~\cite{Konno:2017mos} for higher-rank cases.
}
From the OPE \eqref{eq:OPE_z}, we have
    \begin{align}
        \mathsf{Z}_K^{(m)}(x) z = (K q_1^{-2m} z) \mathsf{Z}_K^{(m)}(x)
        \, .
        \label{eq:Zz_weight}
    \end{align}
The factor $K q_1^{-2m}$ corresponds to the weight of the $m$-th operator:
Under the specialization $K = q_1^{\mu}$ for $\mu \in \mathbb{Z}_{\ge 0}$, we have $K q_1^{-2m} = q_1^{\mu - 2m}$. 
See \eqref{eq:weight_fac} for a related factor.
The parameter $z$ is called the dynamical parameter in the context of elliptic integrable systems, which is also identified with the FI parameter of three-dimensional gauge theory.

Let us consider the conformal block of the screened vertex operators.
Define a map, $c: \{1,\ldots,k\} \to \{1,\ldots,n\}$, that we call the coloring function.
Let $\underline{k} = (k_1,\ldots,k_n) \in (\mathbb{Z}_{\ge 0})^n$, such that $k_\alpha = \#\{ c(i) = \alpha , i = 1,\ldots, k \}$, obeying $\sum_{\alpha=1}^n k_\alpha = k$.
We write
\begin{align}
    K_i^c = \prod_{\alpha \le c(i)} K_\alpha \, ,
\end{align}
and the Weyl vector of type $A_{k-1}$ (a half of the positive root sum) by
\begin{align}
        \rho_i = \frac{k+1}{2} - i \, , \qquad i = 1,\ldots, k \, .
\end{align}
Then, we have the conformal block of the screened vertex operators.

\begin{proposition}\label{prop:conf_block_A1}
    We have the following integral formula for the conformal block in the chamber denoted by $\mathscr{C}_{12\ldots n}$, such that $a_1 \succ a_2 \succ \cdots \succ a_n$ under the parameter shift $z \mapsto z q_1^{k-1}$,
    \begin{align}
        \Psi_{\mathscr{C}_{1\ldots n}}^{\underline{k}}(\underline{a};\underline{K}) & = \bra{0} \mathsf{Z}^{(k_1)}_{K_1}(a_1) \mathsf{Z}^{(k_2)}_{K_2}(a_2) \cdots \mathsf{Z}^{(k_n)}_{K_n}(a_n) \ket{0} / \bra{0} \mathsf{Z}_{K_1}(a_1) \mathsf{Z}_{K_2}(a_2) \cdots \mathsf{Z}_{K_n}(a_n) \ket{0}\Big|_{z \to z q_1^{k-1}}
        \nonumber \\ 
        & = \frac{\mathsf{c}_{\underline{k}}(\underline{K})}{k!} \oint_\gamma [\dd{\underline{x}}] \, \Phi_{2}(\underline{x},\underline{a})         \operatorname{Stab}_{2,\underline{k}}
        (\underline{x},\underline{a})
    \end{align}
    where we define the principal part of the integrand for $\sigma \in \underline{\mathbf{2}}$, $\bar{\sigma} = \underline{\mathbf{2}} \backslash \sigma$ with the measure
    \begin{align}
        \Phi_{\sigma}(\underline{x},\underline{a}) = \prod_{i \neq j}^k \frac{(x_{ij};q_\sigma^{-1})_\infty}{(q_{\bar{\sigma}} x_{ij};q_\sigma^{-1})_\infty} \prod_{\substack{i=1,\ldots,k \\ \alpha=1,\ldots,n}} \frac{(K_\alpha a_\alpha/x_i;q_\sigma^{-1})_\infty}{( a_\alpha/x_i;q_\sigma^{-1})_\infty} 
        \, , \qquad 
        [\dd{\underline{x}}] = \prod_{i=1}^k \frac{\dd{x}_i}{2 \pi \ii x_i}
        \, ,
        \label{eq:principal_A1}
    \end{align}
    together with the constant $\mathsf{c}_{\underline{k}}(\underline{K}) = \mathsf{c}_{2,\underline{k}}(\underline{K})$ defined by
    \begin{align}
        \mathsf{c}_{\sigma,\underline{k}}(\underline{K}) = \prod_{i=1}^k \frac{\theta(K_{c(i)};q_\sigma^{-1})}{\theta(z q_{\bar{\sigma}}^{-2\rho_i} K_i^c;q_\sigma^{-1})} \, . \label{eq:c_const}
    \end{align}
    The off-shell elliptic stable envelope is defined for $\sigma \in \underline{\mathbf{2}}$, $\bar{\sigma} = \underline{\mathbf{2}} \backslash \sigma$ by 
    \begin{align}
        \operatorname{Stab}_{\sigma,\underline{k}}(\underline{x},\underline{a}) 
        & = \operatorname*{Sym}_{\underline{x}} \prod_{i<j}^k \frac{\theta(q_{\bar{\sigma}} x_{ij};q_\sigma^{-1})}{\theta(x_{ij};q_\sigma^{-1})} \prod_{i=1}^k \frac{\prod_{\alpha < c(i)} \theta(a_\alpha / x_i ; q_\sigma^{-1}) \theta(z q_{\bar{\sigma}}^{-2\rho_i} K_i^c a_{c(i)} / x_i ; q_\sigma^{-1}) \prod_{\alpha > c(i)} \theta(K_\alpha a_\alpha / x_i ; q_\sigma^{-1})}{\prod_{\alpha=1}^k \theta(K_\alpha a_\alpha / x_i ; q_\sigma^{-1})} .
        \label{eq:stab_2}
    \end{align}
\end{proposition}
The parameter shift $z \mapsto z q_1^{k-1}$ is applied to all the $z$ parameters appearing in the formula, so that the $q$-shift of the dynamical parameter can be written in terms of the Weyl vector.
\begin{proof}
By Lemma~\ref{lemma:OPE_flip}, we have
\begin{align}
    \Psi_{\mathscr{C}_{1\ldots n}}^{\underline{k}}(\underline{a};\underline{K}) & = \frac{1}{k!} \oint_\gamma \prod_{i=1}^k \frac{\dd{x}_i}{2 \pi \ii x_i} \varphi_{K_{c(i)}}(x_i,a_{c(i)},z q_1^{-2\rho_i} K_i^c) \prod_{i=1,\ldots,k}^{\curvearrowleft} \mathsf{S}(x_i) \prod_{\alpha=1,\ldots,n}^{\curvearrowright} \mathsf{Z}_{K_\alpha}(a_\alpha) 
        \prod_{i=1}^k \prod_{\alpha < c(i)} \frac{\theta(a_\alpha / x_i;q_2^{-1})}{\theta(K_\alpha a_\alpha / x_i;q_2^{-1})}
\end{align}
and
\begin{align}
    & \prod_{i=1}^k \varphi_{K_{c(i)}}(x_i,a_{c(i)}, z q_1^{-2\rho_i} K_i^c) \prod_{\alpha < c(i)} \frac{\theta(a_\alpha / x_i;q_2^{-1})}{\theta(K_\alpha a_\alpha / x_i;q_2^{-1})}
    \nonumber \\
    & = \prod_{i=1}^k \frac{\theta(K_{c(i)};q_2^{-1})}{\theta(z q_1^{-2\rho_i} K_i^c;q_2^{-1})} \frac{\prod_{\alpha < c(i)} \theta(a_\alpha / x_i ; q_2^{-1}) \theta(z q_1^{-2\rho_i} K_i^c a_{c(i)} / x_i ; q_2^{-1}) \prod_{\alpha > c(i)} \theta(K_\alpha a_\alpha / x_i ; q_2^{-1})}{\prod_{\alpha=1}^k \theta(K_\alpha a_\alpha / x_i ; q_2^{-1})}
\end{align}
Then, using the OPEs shown in Definition~\ref{def:OPEs}, we obtain the expression above.
\end{proof}

In the following, we drop the constant factor $\mathsf{c}_{\underline{k}}(\underline{K})$ for simplicity.

\begin{remark}
    The factor $K_i^c$ is different from the standard convention of the elliptic stable envelope by the factor $K_{c(i)}$. 
    We apply this modification to obtain a consistent result with the known expressions of the vortex partition function. See Proposition~\ref{prop:vort_fn}.
\end{remark}

The conformal block $\Psi_{\mathscr{C}_{1\ldots n}}^{\underline{k}}(\underline{a};\underline{K})$ of W$_{q_{1,2}}(A_1)$ algebra agrees 
with the partition function of the mass-deformed $\mathcal{N}=4$ three-dimensional $\mathrm{U}(k)$ gauge theory supported on $\mathbb{C}_{q_2} \times \mathbb{S}^1$ with $n$ fundamental matters~\cite{Aganagic:2013tta,Aganagic:2014oia,Aganagic:2015cta} together with the additional boundary contribution written in terms of the theta functions.%
\footnote{%
Meanwhile, the partition function of $\mathcal{N}=1$ five-dimensional gauge theory defined on $\mathbb{C}_{q_1} \times \mathbb{C}_{q_2} \times \mathbb{S}^1$ is similarly realized by the conformal block of the (infinitely many) screening charges~\cite{Kimura:2015rgi}.
}
In terms of the $\mathcal{N}=2$ multiplets, it consists of a vector multiplet, a single adjoint and $(n+n)$ fundamental matters, which is known as the handsaw quiver gauge theory of type $A_1$.
The corresponding quiver diagram is given in Fig.~\ref{fig:single_A1}.

\paragraph{Pole structure}
In order to evaluate the contour integral, we discuss the pole structure of the integrand.
Let $x_{\alpha,i} = x_{i+\sum_{\beta > \alpha}k_\beta}$ for $i = 1,\ldots, k_\alpha$.
The poles of the contour integral are parametrized by $n$-tuple reverse partitions,
\begin{align}
    x_{\alpha,1} = a_\alpha q_2^{-n_{\alpha,1}} 
    \, , \quad 
    x_{\alpha,i} = x_{\alpha,i-1} q_1 q_2^{-n_{\alpha,i}} = a_\alpha q_1^{i-1} q_2^{-\nu_{\alpha,i}}
\end{align}
where $n_{\alpha,i} \in \mathbb{Z}_{\ge 0}$ and $\nu_\alpha = (\nu_{\alpha,1} \le \nu_{\alpha,2} \le \cdots \le \nu_{\alpha,k_\alpha})$ is a reverse partition defined by
\begin{align}
    \nu_{\alpha,i} = \sum_{j=1}^i n_{\alpha,j} \, .
\end{align}
We see that the residues of other poles at $x_{\alpha,i}/x_{\alpha,j} = q_1 q_2^{\mathbb{Z}_{\le 0}}$ for $i < j$ are zero due to the stable envelope that plays a role of the pole subtraction matrix~\cite{Aganagic:2016jmx}.
\begin{remark}
Although we can convert a reverse partition to the standard one by relabeling the indices, $i \mapsto k_\alpha - i + 1$, we now apply this convention, which is naturally generalized to the situation with the double screening current (See \S\ref{sec:double_screening_A1}).
This convention is also compatible with the 5d gauge theory partition function (Proposition~\ref{prop:Z_5d}).
\end{remark}

\begin{proposition}\label{prop:Higgs1}
    Let $K_\alpha = q_1^{\mu_\alpha}$, $\mu_\alpha \in \mathbb{Z}_{\ge 0}$ for $\alpha = 1,\ldots,n$.
    Then, $\Psi^{\underline{k}} = 0$ unless $0 \le k_\alpha \le \mu_\alpha$.
\end{proposition}
\begin{proof}
    The factor in the integrand $(K_\alpha a_\alpha / x_{\alpha,i}; q_2^{-1})_\infty = (q_1^{\mu_\alpha} a_\alpha / x_{\alpha,i}; q_2^{-1})_\infty$ has zeros at $x_{\alpha,i} = a_\alpha q_1^{\mu_\alpha} q_2^{\mathbb{Z}_{\le 0}}$.
\end{proof}
This truncation is essential to realize the finite-dimensional representation of the quantum affine algebra. 
In this context, the simplest case $n = 1$, $\mu = 1$ ($K = q_1$) corresponds to the fundamental two-dimensional representation of $U_q(\widehat{\mathfrak{sl}}_2)$, and a general case is obtained from the tensor product of the fundamental ones via specialization of the parameters, $\{a_1,\ldots,a_\mu\} = \{a, a q_1, \ldots, a q_1^{\mu-1} \}$.
See \eqref{eq:Zz_weight}.
Under this specialization, we have
\begin{align}
    \prod_{\alpha = 1}^\mu \frac{(q_1 a_\alpha/x;q_2^{-1})}{(a_\alpha/x;q_2^{-1})} = \frac{\cancel{(q_1 a /x;q_2^{-1})}}{(a /x;q_2^{-1})} \frac{\cancel{(q_1^2 a /x;q_2^{-1})}}{\cancel{(q_1 a /x;q_2^{-1})}} \cdots \frac{(q_1^{\mu} a /x;q_2^{-1})}{\cancel{(q_1^{\mu-1} a /x;q_2^{-1})}} = \frac{(q_1^\mu a /x;q_2^{-1})}{(a /x;q_2^{-1})} \, ,
\end{align}
which is the case $K = q_1^\mu$.
This condition is known as the $q$-segment condition for (the roots of) the Drinfeld polynomial, which characterizes the finite-dimensional representation of the quantum affine algebra (and also the Yangian)~\cite{Chari:1991CMP}. 
From the geometric point of view, the assignment $K = q_1^\mu$ is realized via modification of the superpotential in the corresponding quiver variety description, $W = \sum_{\alpha = 1}^n J_\alpha B^{\mu_\alpha} I_\alpha$~\cite{Nekrasov:2009uh,Bykov:2019cst,Cao:2023lon}.
Hence, denoting by $V_\alpha$ the ($\mu_\alpha+1$)-dimensional module with the decomposition, $V_\alpha = \oplus_{k=0}^{\mu_\alpha} V_{\alpha;k}$, and applying the specialization $K_\alpha = q_1^{\mu_\alpha}$ as in Proposition~\ref{prop:Higgs1} for $\alpha = 1,\ldots,n$, the conformal block takes a value in the tensor product module,
\begin{align}
    \Psi_{\mathscr{C}_{1\ldots n}}(\underline{a};\underline{K}) = \sum_{\underline{k}} (V_{n;k_n} \otimes \cdots \otimes V_{1;k_1}) \Psi_{\mathscr{C}_{1\ldots n}}^{\underline{k}}(\underline{a};\underline{K}) \, .
\end{align}

\subsubsection{Elliptic $R$-matrix}

Let us consider the simplest non-trivial example $(n,k) = (2,1)$, where the corresponding quiver variety is $T^\vee\mathbb{P}^1$ from the geometric point of view.
Hence, the conformal block computes the vertex function associated with quasimaps $f: \mathbb{P}^1 \to X = T^\vee\mathbb{P}^1$.
We write $K_{i_1 i_2 \cdots i_m} = K_{i_1} K_{i_2} \cdots K_{i_m}$ and we write $(\cdot)_d = (\cdot;q_2^{-1})_d$, $\theta(\cdot) = \theta(\cdot;q_2^{-1})$.
Then, we have two possible conformal blocks, $k = (1,0)$ and $(0,1)$, in the chamber ${\mathscr{C}_{12}} : a_1 \succ a_2$, 
\begin{subequations}\label{eq:TP1_block}
\begin{align}
    \Psi_{\mathscr{C}_{12}}^{(1,0)}(a_{1,2};K_{1,2}) & = \oint_\gamma \frac{\dd{x}}{2 \pi \ii x} \frac{(K_1 a_1/x,K_2 a_2/x)_\infty}{(a_1/x,a_2/x)_\infty} \frac{\theta(z K_1 a_1/x) \theta(K_1)}{\theta(K_1 a_1/x) \theta(z K_1)}
    \, , \\ 
    \Psi_{\mathscr{C}_{12}}^{(0,1)}(a_{1,2};K_{1,2}) & = \oint_\gamma \frac{\dd{x}}{2 \pi \ii x} \frac{(K_1 a_1/x,K_2 a_2/x)_\infty}{(a_1/x,a_2/x)_\infty} \frac{\theta(z K_{12} a_2/x) \theta(K_2)}{\theta(K_2 a_2/x) \theta(z K_{12})} \frac{\theta(a_1/x)}{\theta(K_1 a_1 / x)}
    \, .
\end{align}
\end{subequations}
Taking the poles at $x = a_1 q_2^{\mathbb{Z}_{\le 0}}$, we have
\begin{subequations}
\begin{align}
    \Psi_{\mathscr{C}_{12}}^{(1,0)}\Big|_{a_1 q_2^{\mathbb{Z}_{\le 0}}} 
    & = \frac{(K_1,K_2 a_{21})_\infty}{(q_2^{-1},a_{21})_\infty} {_2\phi_1}\left[ 
    \begin{matrix}
        K_1^{-1}q_2^{-1} & K_2^{-1} a_{12} q_2^{-1} \\ a_{12} q_2^{-1} & -
    \end{matrix} ; q_2^{-1}; z K_{12} \right]
    \, , \\ 
    \Psi_{\mathscr{C}_{12}}^{(0,1)}\Big|_{a_1 q_2^{\mathbb{Z}_{\le 0}}} & = 0 \, ,
\end{align}
\end{subequations}    
where ${}_2\phi_1$ is the $q$-hypergeometric series (Definition~\ref{def:q-hypergeometric_fn}), and for the poles at $x=a_2 q_2^{\mathbb{Z}_{\le 0}}$, we have
\begin{subequations}
\begin{align}
    \Psi_{\mathscr{C}_{12}}^{(1,0)}\Big|_{a_2 q_2^{\mathbb{Z}_{\le 0}}} 
    & = \frac{(K_2,q_2^{-1} a_{21})_\infty}{(q_2^{-1},q_2^{-1} K_1^{-1} a_{21})_\infty} \frac{\theta(z K_1 a_{12},K_1)}{\theta(a_{12},z K_1)} {_2\phi_1}\left[ 
    \begin{matrix}
        K_2^{-1}q_2^{-1} & K_1^{-1} a_{21} q_2^{-1} \\ a_{21} q_2^{-1} & -
    \end{matrix} ; q_2^{-1}; z K_{12} \right]   
    \, , \\
    \Psi_{\mathscr{C}_{12}}^{(0,1)}\Big|_{a_2 q_2^{\mathbb{Z}_{\le 0}}} 
    & = \frac{(K_2,K_1 a_{12})_\infty}{(q_2^{-1},a_{12})_\infty} \frac{\theta(a_{12}) }{\theta(K_1 a_{12}) }{_2\phi_1}\left[ 
    \begin{matrix}
        K_2^{-1}q_2^{-1} & K_1^{-1} a_{21} q_2^{-1} \\ a_{21} q_2^{-1} & -
    \end{matrix} ; q_2^{-1}; z K_{12} \right]      
    \, ,
\end{align}
\end{subequations}
where we assumed $|z K_{12}| < 1$.
Combining them together, we write
\begin{align}
    \Psi_{\mathscr{C}_{12}} & =
    \begin{pmatrix}
        \Psi_{\mathscr{C}_{12}}^{(1,0)}\Big|_{a_1 q_2^{\mathbb{Z}_{\le 0}}} + \Psi_{\mathscr{C}_{12}}^{(1,0)}\Big|_{a_2 q_2^{\mathbb{Z}_{\le 0}}} \\ \Psi_{\mathscr{C}_{12}}^{(0,1)}\Big|_{a_1 q_2^{\mathbb{Z}_{\le 0}}} + \Psi_{\mathscr{C}_{12}}^{(0,1)}\Big|_{a_2 q_2^{\mathbb{Z}_{\le 0}}}
    \end{pmatrix}
    \nonumber \\
    & =
    \begin{pmatrix}
            1 & \displaystyle \frac{\theta(z K_1 a_{12},K_1)}{\theta(K_1 a_{12},z K_1) } \\ 
            0 & \displaystyle \frac{\theta(a_{12}) }{\theta(K_1 a_{12}) }
    \end{pmatrix}
    \begin{pmatrix}
        \displaystyle \frac{(K_1,K_2 a_{21})_\infty}{(q_2^{-1},a_{21})_\infty} {_2\phi_1}\left[ 
    \begin{matrix}
        K_1^{-1}q_2^{-1} & K_2^{-1} a_{12} q_2^{-1} \\ a_{12} q_2^{-1} & -
    \end{matrix} ; q_2^{-1}; z K_{12} \right]
    \\[1.5em] \displaystyle
    \frac{(K_2,K_1 a_{12})_\infty}{(q_2^{-1},a_{12})_\infty} {_2\phi_1}\left[ 
    \begin{matrix}
        K_2^{-1}q_2^{-1} & K_1^{-1} a_{21} q_2^{-1} \\ a_{21} q_2^{-1} & -
    \end{matrix} ; q_2^{-1}; z K_{12} \right]
    \end{pmatrix}    
    \, .
\end{align}
In the other chamber $\mathscr{C}_{21} : a_1 \prec a_2$, we instead have
\begin{align}
    \Psi_{\mathscr{C}_{21}} 
    & = 
    \begin{pmatrix}
            \displaystyle \frac{\theta(a_{21}) }{\theta(K_2 a_{21}) } & 0 \\
            \displaystyle \frac{\theta(z K_2 a_{21},K_2)}{\theta(K_2 a_{21},z K_2) } & 1
    \end{pmatrix}
    \begin{pmatrix}
        \displaystyle \frac{(K_1,K_2 a_{21})_\infty}{(q_2^{-1},a_{21})_\infty} {_2\phi_1}\left[ 
    \begin{matrix}
        K_1^{-1}q_2^{-1} & K_2^{-1} a_{12} q_2^{-1} \\ a_{12} q_2^{-1} & -
    \end{matrix} ; q_2^{-1}; z K_{12} \right]
    \\[1.5em] \displaystyle
    \frac{(K_2,K_1 a_{12})_\infty}{(q_2^{-1},a_{12})_\infty} {_2\phi_1}\left[ 
    \begin{matrix}
        K_2^{-1}q_2^{-1} & K_1^{-1} a_{21} q_2^{-1} \\ a_{21} q_2^{-1} & -
    \end{matrix} ; q_2^{-1}; z K_{12} \right]
    \end{pmatrix}
    \, .
\end{align}
Change of the chambers is simply organized by the $R$-matrix constructed from the stable envelope.
\begin{proposition}\label{prop:R-matrix}
    \begin{align}
        \Psi_{\mathscr{C}_{12}}(a_{1,2};K_{1,2}) = R_{21} \Psi_{\mathscr{C}_{21}}(a_{1,2};K_{1,2})
    \end{align}
    where the elliptic $R$-matrix is given by the elliptic stable envelope, 
    \begin{align}
        R_{21} = R_{21 \to 12} = 
        \begin{pmatrix}
            1 & \displaystyle \frac{\theta(z K_1 a_{12},K_1)}{\theta(K_1 a_{12},z K_1) } \\ 
            0 & \displaystyle \frac{\theta(a_{12}) }{\theta(K_1 a_{12}) }
        \end{pmatrix}
        \begin{pmatrix}
            \displaystyle \frac{\theta(a_{21}) }{\theta(K_2 a_{21}) } & 0 \\
            \displaystyle \frac{\theta(z K_2 a_{21},K_2)}{\theta(K_2 a_{21},z K_2) } & 1
        \end{pmatrix}^{-1}
        \, , \quad 
        R_{12} = R_{12 \to 21} = R_{21 \to 12}^{-1} \, .
    \end{align}
\end{proposition}
We may use the addition formula of the theta functions to rewrite the elliptic $R$-matrix into the original form by Felder~\cite{Felder:1994pb,Felder:1994be} together with the gauge transform~\cite{Aganagic:2016jmx}.
For the $n$-point conformal block, in general, we have the $R$-matrix acting on the $\alpha$-th and $\beta$-th modules,
\begin{align}
    \Psi_{\mathscr{C}_{\ldots \beta \ldots \alpha \ldots}}(\underline{a};\underline{K}) = R_{\alpha \beta} \Psi_{\mathscr{C}_{\ldots \alpha \ldots \beta \ldots}}(\underline{a};\underline{K}) \, .
\end{align}

\paragraph{Mirror symmetry}

By Proposition~\ref{prop:Heine_transform}, we may rewrite the principal part of the conformal block as follows, 
\begin{align}
    &
    \begin{pmatrix}
        \displaystyle \frac{(K_1,K_2 a_{21})_\infty}{(q_2^{-1},a_{21})_\infty} {_2\phi_1}\left[ 
    \begin{matrix}
        K_1^{-1}q_2^{-1} & K_2^{-1} a_{12} q_2^{-1} \\ a_{12} q_2^{-1} & -
    \end{matrix} ; q_2^{-1}; z K_{12} \right]
    \\[1.5em] \displaystyle
    \frac{(K_2,K_1 a_{12})_\infty}{(q_2^{-1},a_{12})_\infty} {_2\phi_1}\left[ 
    \begin{matrix}
        K_2^{-1}q_2^{-1} & K_1^{-1} a_{21} q_2^{-1} \\ a_{21} q_2^{-1} & -
    \end{matrix} ; q_2^{-1}; z K_{12} \right]
    \end{pmatrix} 
    \nonumber \\
    & =
    \begin{pmatrix}
        \displaystyle
        \frac{\theta(K_2 a_{21})}{\theta(a_{21})} \frac{(K_1,z K_2 q_2^{-1})_\infty}{(q_2^{-1},zK_{12})_\infty}
        {_2\phi_1} \left[
        \begin{matrix}
            K_2 & z K_{12} \\ z K_2 q_2^{-1}
        \end{matrix} ; q_2^{-1} ; \frac{a_{12}}{q_2 K_2}
        \right] \\[1.5em]
        \displaystyle
        \frac{\theta(K_1 a_{12})}{\theta(a_{12})} \frac{(K_2,z K_1 q_2^{-1})_\infty}{(q_2^{-1},zK_{12})_\infty}
        {_2\phi_1} \left[
        \begin{matrix}
            K_1 & z K_{12} \\ z K_1 q_2^{-1}
        \end{matrix} ; q_2^{-1} ; \frac{a_{21}}{q_2 K_1}
        \right]
    \end{pmatrix}
    \, .
\end{align}
The $q$-hypergeometric series on the LHS are analytic when $|z K_{12}| < 1$, while it is not clear if they are all analytic in the chamber ${\mathscr{C}_{12}}: a_1 \succ a_2$, i.e., $|a_{21}| > 1$.
Meanwhile, we may also write the total conformal block (including the stable envelope contribution),
\begin{align}
    \Psi_{\mathscr{C}_{12}}
    & =
    \begin{pmatrix}
        1 & \displaystyle \frac{\theta(z K_1 a_{12},K_1)}{\theta(a_{12},zK_1)} \\ 0 & 1
    \end{pmatrix}
    \begin{pmatrix}
        \displaystyle
        \frac{(z q_2^{-1}, K_1)_\infty}{(q_2^{-1}, z K_1)_\infty}
    {_2 \phi_1} \left[
    \begin{matrix}
        K_2 & z^{-1} \\ z^{-1} K_1^{-1} q_2^{-1}
    \end{matrix} ; q_2^{-1} ; \frac{a_{21}}{q_2 K_1}
    \right] \\[1.5em]
    \displaystyle
    \frac{(K_2,z K_1 q_2^{-1})_\infty}{(q_2^{-1},z K_{12})_\infty } {_2\phi_1}
    \left[
    \begin{matrix}
        K_1 & z K_{12} \\ z K_1 q_2^{-1} &
    \end{matrix}; q_2^{-1}; \frac{a_{21}}{q_2 K_1}
    \right] 
    \end{pmatrix}
    \, ,
\end{align}
which is now clearly analytic in the chamber ${\mathscr{C}_{12}}$ through the basis change organized by the stable envelope, implying the 3d mirror symmetry that exchanges the twisted masses $a_{12}$ ($a_{21}$) and the FI parameter $z$~\cite{Bullimore:2021rnr,Dedushenko:2021mds}.
From this point of view, the $R$-matrix organizes change of the chamber $a_{12}$, $z \gtrless 1$, i.e., the wall-crossing.

\subsubsection{Vertex function}\label{sec:vort_part_fn}

Let $\underline{\ell} = (\ell_1,\ldots,\ell_n) \in (\mathbb{Z}_{\ge 0})^n$ obeying $\sum_{\alpha=1}^n \ell_\alpha = k$ and the corresponding coloring function $c' = c_{\underline{\ell}}$, such that $\ell_\alpha = \#\{ c'(i) = \alpha, i = 1,\ldots,k \}$.
We denote by $\mathsf{P}_{\underline{\ell}}$ the set of $n$-tuple reverse partitions $\underline{\nu} = (\nu_1,\ldots,\nu_\alpha)$ obeying $\nu_\alpha = (\nu_{\alpha,1} \le \cdots \le \nu_{\alpha,\ell_\alpha})$.
We consider the residue of the following poles,
\begin{align}
    \underline{x}(\underline{\nu}) = \{ x_{\alpha,i} \mid x_{\alpha,i} = x_{i + \sum_{\beta > \alpha} \ell_\beta} = a_\alpha q_1^{i-1} q_2^{-\nu_{\alpha,i}}, \alpha = 1,\ldots, n, i = 1, \ldots,\ell_\alpha \} \, , \quad \underline{\nu} \in \mathsf{P}_{\underline{\ell}} \, .
    \label{eq:pole_ell}
\end{align}

\begin{lemma}
    The residue of the integrand of $\Psi_{\mathscr{C}_{1\ldots n}}^{\underline{k}}$ at \eqref{eq:pole_ell} is zero unless $c'(i) \ge c(i)$ for $i = 1 ,\ldots,k$.
\end{lemma}
This implies the triangular property of the elliptic stable envelope.
In order to write down the vertex function, we need the following Lemma.
\begin{lemma}\label{lemma:A1_pert}
    Let $\emptyset_{\underline{k}} \in \mathsf{P}_{\underline{k}}$ be the empty configuration and $(\cdot)_\infty = (\cdot;q_2^{-1})_{\infty}$.
    Denoting $k_{\beta\alpha} = k_\beta - k_\alpha$, we have
\begin{align}
    I_{\underline{k}} (\underline{a};\underline{K}) 
    & = \frac{(q_2^{-1})_\infty^k}{k!} \oint_{\underline{x}(\emptyset_{\underline{k}})} [\dd{\underline{x}}] \frac{\prod_{i \neq j}^k (x_{ij})_\infty}{\prod_{i,j}^k (q_1 x_{ij})_\infty} \prod_{\substack{\alpha = 1,\ldots,n \\ i = 1, \ldots, k}} \frac{(K_\alpha a_\alpha / x_i)_\infty}{(a_\alpha / x_i)_\infty}
    = \prod_{\substack{\alpha,\beta=1,\ldots,n \\ i = 1,\ldots,k_\alpha}} \frac{(K_\beta a_{\beta\alpha} q_1^{1-i})_\infty}{(a_{\beta\alpha} q_1^{k_{\beta\alpha}+i})_\infty} \, .
\end{align}
\end{lemma}
\begin{proof}
    Let $\mathbf{K}$, $\mathbf{N}$, and $\widetilde{\mathbf{N}}$ be vector spaces (vector bundles whose fiber is given by these vector spaces on the moduli space of quasimaps, $\mathsf{QM}(T^\vee\mathbb{P}^1) = \{f : \mathbb{P}^1 \to X = T^\vee\mathbb{P}^1\}/\sim$) with the Chern characters, $\operatorname{ch} \mathbf{K} = \sum_{i=1}^k x_i$, $\operatorname{ch} \mathbf{N} = \sum_{\alpha=1}^n a_\alpha$, and $\operatorname{ch} \widetilde{\mathbf{N}} = \sum_{\alpha=1}^n K_\alpha a_\alpha$.
    We define $\mathbf{P}_a$ with $\operatorname{ch} \mathbf{P}_a = 1 - q_a$.
    For the dual space, we have $\operatorname{ch} \mathbf{K}^\vee = \sum_{i=1}^k x_i^{-1}$, $\operatorname{ch} \mathbf{P}_a^\vee = 1 - q_a^{-1}$, and we denote by $\wedge \mathbf{X} = \sum_{i \ge 0} (-1)^i \wedge^i \mathbf{X}$ the alternating anti-symmetrized sum.
    Then, the integral is identified with the contribution of the fixed point $\emptyset_{\underline{k}}$ of the equivariant index (equivariant K-theoretic Euler class),
    \begin{align}
    I_{\underline{k}} (\underline{a};\underline{K}) = \operatorname{ch} \wedge \left[ \frac{1}{\mathbf{P}_2^\vee} \left( \mathbf{P}_1 \mathbf{K} \mathbf{K}^\vee + \widetilde{\mathbf{N}} \mathbf{K}^\vee - \mathbf{N} \mathbf{K}^\vee \right) \right]_{\emptyset_{\underline{k}}} \, .
    \end{align}    
    Since we have (see, e.g.,~\cite{Shadchin:2006yz} for details)
    \begin{align}
        \operatorname{ch} \left( \mathbf{P}_1 \mathbf{K} \mathbf{K}^\vee + \widetilde{\mathbf{N}} \mathbf{K}^\vee - \mathbf{N} \mathbf{K}^\vee \right)_{\emptyset_{\underline{k}}} = - \sum_{\substack{\alpha,\beta = 1, \ldots, n \\ i = 1, \ldots, k_\alpha}} a_{\beta \alpha} q_1^{k_\beta - k_\alpha + i} \, ,
    \end{align}
    we obtain the expression above by evaluating the index.
\end{proof}

We now write down the vertex function in general.
\begin{proposition}\label{prop:vort_fn}
We write $(\cdot)_d = (\cdot;q_2^{-1})_d$ and $\theta(\cdot) = \theta(\cdot;q_2^{-1})$.
Let $I = i + \sum_{\alpha' > \alpha} \ell_{\alpha'}$, $J = j + \sum_{\alpha' > \beta} \ell_{\alpha'}$, hence $c'(I) = \alpha$ and $c'(J) = \beta$.
We also abuse the notation, $I = (\alpha,i)$ and $J = (\beta,j)$, to write $\nu_I = \nu_{\alpha,i}$ and $\nu_J = \nu_{\beta,j}$ for simplicity.
    Then, the contribution of the conformal block $\Psi_{\mathscr{C}_{1\ldots n}}^{\underline{k}}(\underline{a};\underline{K})$ associated with the poles \eqref{eq:pole_ell} is given as follows,
\begin{align}
    \left.\Psi_{\mathscr{C}_{1\ldots n}}^{\underline{k}}(\underline{a};\underline{K})\right|_{\underline{\ell}} = \operatorname{Stab}_{\underline{k};\underline{\ell}}(\underline{a}) \mathring{Z}_{\underline{\ell}} \sum_{\underline{\nu} \in \mathsf{P}_{\underline{\ell}}} z^{|\underline{\nu}|} {Z}_{\underline{\nu}}
    \, .
\end{align}
We have the so-called perturbative part and the vortex part,
\begin{subequations}
\begin{align}
    \mathring{Z}_{\underline{\ell}}
    & = 
    \frac{(q_1)^k}{(q_2^{-1})^k} \prod_{\substack{\alpha,\beta=1,\ldots,n \\ i = 1,\ldots,\ell_\alpha}} \frac{(K_\beta a_{\beta\alpha} q_1^{1-i})_\infty}{(a_{\beta\alpha} q_1^{\ell_{\beta\alpha}+i})_\infty}
    \, , \\
    Z_{\underline{\nu}} 
    & = \prod_{(\alpha,i) \neq (\beta,j)} 
    \frac{(q_1^{j-i+1} a_{\beta\alpha})_{\nu_{\beta,j}-\nu_{\alpha,i}}}{(q_1^{j-i} a_{\beta\alpha})_{\nu_{\beta,j}-\nu_{\alpha,i}}} \prod_{\alpha,\beta=1}^n \prod_{i=1}^{\ell_\alpha} \frac{(q_1^{1-i} a_{\beta \alpha})_{-\nu_{\alpha,i}}}{(q_1^{1-i} K_\beta a_{\beta \alpha})_{-\nu_{\alpha,i}}}
    \, ,
\end{align}
\end{subequations}
    which agree with the vertex function for $X = T^\vee\!\operatorname{Gr}(k,n)$ (Nakajima quiver variety associated with $A_1$ quiver).
    The specialization of the elliptic stable envelope (the on-shell stable envelope) is given by
\begin{align}
    \operatorname{Stab}_{\underline{k};\underline{\ell}}(\underline{a}) & := \operatorname{Stab}_{2,\underline{k}}(x_i = a_{c'(i)},\underline{a})
    \nonumber \\
    & = \prod_{1 \le I < J \le k} \frac{\theta(q_1^{j-i+1} a_{\beta\alpha})}{\theta(q_1^{j-i} a_{\beta\alpha})} \prod_{I=1}^k 
    \left(
    \prod_{\beta < \alpha} \frac{\theta(q_1^{1-i} a_{\beta\alpha})}{\theta(q_1^{1-i} K_\beta a_{\beta\alpha})}
    \right)
    \frac{\theta(q_1^{1-i-2\rho_I} z K_{1\ldots \alpha}) \theta(K_{\alpha})}{\theta(q_1^{1-i} K_{\alpha}) \theta(q_1^{-2\rho_I} z K_{1\ldots \alpha})}
    \, .
\end{align}

\end{proposition}
\begin{proof}
The perturbative part can be directly obtained by Lemma~\ref{lemma:A1_pert}.
It is straightforward to evaluate the residue to obtain the vortex contributions,
\begin{align}
    \frac{1}{k!} \oint_{\underline{x}(\underline{\nu})} [\dd{\underline{x}}] \, \Phi_{2}(\underline{x},\underline{a})
    & = \mathring{Z}_{\underline{\ell}} Z_{\underline{\nu}}  
    \, .
\end{align}
We then consider the stable envelope.
We have
\begin{align}
    \left.\prod_{1 \le I<J \le k} \frac{\theta(q_1 x_{JI})}{\theta(x_{JI})}\right|_{\underline{x}(\underline{\nu})}
    & = \prod_{1 \le I<J \le k} \frac{\theta(q_2^{\nu_{I} - \nu_{J}} q_1^{j-i+1} a_{\beta\alpha})}{\theta(q_2^{\nu_{I} - \nu_{J}} q_1^{j-i} a_{\beta\alpha})} = \prod_{1 \le I < J \le k} q_1^{\nu_{I} - \nu_{J}} \frac{\theta(q_1^{j-i+1} a_{\beta\alpha})}{\theta(q_1^{j-i} a_{\beta\alpha})}
    \nonumber \\
    & = 
    q_1^{\sum_{I=1}^k 2 \rho_I \nu_I} \prod_{1 \le I < J \le k} \frac{\theta(q_1^{j-i+1} a_{\beta\alpha})}{\theta(q_1^{j-i} a_{\beta\alpha})}
    \, ,
\end{align}
and 
\begin{align}
    &
    \left.\prod_{I=1}^k \left(
    \prod_{\beta < \alpha} \frac{\theta(a_\beta/x_I)}{\theta(K_\beta a_\beta/x_I)}
    \right)
    \frac{\theta(z q_1^{-2\rho_I} K_{1\ldots \alpha} a_{\alpha} / x_I) \theta(K_{\alpha})}{\theta(K_{\alpha} a_{\alpha} / x_I) \theta(z q_1^{-2\rho_I} K_{1\ldots \alpha})}
    \right|_{\underline{x}(\underline{\nu})}
    \nonumber \\
    & =
    \prod_{I=1}^k \left( \prod_{\beta < \alpha} \frac{\theta(q_2^{\nu_{I}} q_1^{1-i} a_{\beta\alpha})}{\theta(q_2^{\nu_{I}} q_1^{1-i} K_\beta a_{\beta\alpha})}
    \right)
    \frac{\theta(q_2^{\nu_{I}} q_1^{1-i-2\rho_I} z K_{1\ldots \alpha}) \theta(K_{\alpha})}{\theta(q_2^{\nu_{I}} q_1^{1-i} K_{\alpha}) \theta(q_1^{-2\rho_I} z K_{1\ldots \alpha})}
    \nonumber \\
    & =
    \prod_{I=1}^k \left( z q_1^{-2\rho_I} \right)^{\nu_{I}}
    \left(
    \prod_{\beta < \alpha} \frac{\theta(q_1^{1-i} a_{\beta\alpha})}{\theta(q_1^{1-i} K_\beta a_{\beta\alpha})}
    \right)
    \frac{\theta(q_1^{1-i-2\rho_I} z K_{1\ldots \alpha}) \theta(K_{\alpha})}{\theta(q_1^{1-i} K_{\alpha}) \theta(q_1^{-2\rho_I} z K_{1\ldots \alpha})}
    \, .
\end{align}
Hence, the specialization of the stable envelope yields
\begin{align}
    \operatorname{Stab}_{2,\underline{k}}(\underline{x},\underline{a})\Big|_{\underline{x}(\underline{\nu})} & = z^{\sum_{I=1}^k \nu_I} \prod_{1 \le I < J \le k} \frac{\theta(q_1^{j-i+1} a_{\beta\alpha})}{\theta(q_1^{j-i} a_{\beta\alpha})} \prod_{I=1}^k 
    \left(
    \prod_{\beta < \alpha} \frac{\theta(q_1^{1-i} a_{\beta\alpha})}{\theta(q_1^{1-i} K_\beta a_{\beta\alpha})}
    \right)
    \frac{\theta(q_1^{1-i-2\rho_I} z K_{1\ldots \alpha}) \theta(K_{\alpha})}{\theta(q_1^{1-i} K_{\alpha}) \theta(q_1^{-2\rho_I} z K_{1\ldots \alpha})}
    \nonumber \\
    & = z^{|\underline{\nu}|} \operatorname{Stab}_{\underline{k};\underline{\ell}}(\underline{a})
    \, .
\end{align}
This completes the proof.
\end{proof}

\paragraph{5d gauge theory partition function}

We compare the vertex function with 5d gauge theory partition function (K-theoretic Nekrasov partition function).
Let $\underline{\lambda} = (\lambda_\alpha)_{\alpha=1,\ldots,n} = (\lambda_{\alpha,1} \ge \lambda_{\alpha,2} \ge \cdots )_{\substack{\alpha=1,\ldots,n \\ i = 1,\ldots, \infty}}$ be an $n$-tuple partition.
The full partition function of 5d $\mathcal{N}=1$ $\mathrm{U}(n)$ gauge theory on $\mathbb{C}^2 \times \mathbb{S}^1$ with $n$ fundamental and $n$ antifundamental matters is given by%
\footnote{%
Compared with \cite{Kimura:2015rgi,Kimura:2020jxl}, the $q$-parameters are flipped, $q_{1,2} \leftrightarrow q_{1,2}^{-1}$, and slightly modify the notation.
}
\begin{align}
    Z^\text{5d} = \sum_{\underline{\lambda}} \mathfrak{q}^{|\underline{\lambda}|} Z^\text{5d}_{\underline{\lambda}}
    \, ,
\end{align}
where
\begin{align}
    Z^\text{5d}_{\underline{\lambda}} = Z^\text{5d}_{\underline{\lambda}}(\underline{b},\underline{m},\underline{\widetilde{m}};q_{1,2}) = \prod_{(\alpha,i) \neq (\beta, j)} \frac{(b_{\beta\alpha} q_2^{\lambda_{\alpha,i} - \lambda_{\beta,j}} q_1^{i-j}; q_2^{-1})_\infty}{(b_{\beta\alpha} q_2^{\lambda_{\alpha,i} - \lambda_{\beta,j}} q_1^{i-j+1}; q_2^{-1})_\infty} 
    \prod_{\substack{\alpha, f=1,\ldots,n \\ i = 1,\ldots,\infty}} \frac{(\frac{\widetilde{m}_f}{b_\alpha} q_2^{\lambda_{\alpha,i}} q_1^{i};q_2^{-1})_\infty}{(\frac{m_f}{b_\alpha} q_2^{\lambda_{\alpha,i}} q_1^{i};q_2^{-1})_\infty} \, .
\end{align}
We denote the instanton fugacity by $\mathfrak{q}$, the multiplicative Coulomb moduli by $\underline{b} = \{ b_\alpha \}_{\alpha = 1,\ldots,n}$, the mass parameters by $\underline{m} = \{ m_f \}_{f=1,\ldots,n}$ and $\underline{\widetilde{m}} = \{ \widetilde{m}_f \}_{f=1,\ldots,n}$.

\begin{proposition}\label{prop:Z_5d}
    The 5d gauge theory partition function agrees with the vertex function under the identification (the root of Higgs branch) as follows,
    \begin{align}
        \mathfrak{q} = z
        \, , \quad 
        b_\alpha = \widetilde{m}_\alpha = a_\alpha q_1^{-\ell_\alpha}
        \, , \quad 
        m_\alpha = a_\alpha 
        \, , \quad
        K_\alpha = \frac{m_\alpha}{\widetilde{m}_\alpha} = q_1^{\ell_\alpha}
        \, , \quad 
        \lambda_{\alpha,i} = \nu_{\alpha,\ell_\alpha - i + 1}
        \, , \quad 
        (\alpha = 1, \ldots, n) \, .
        \label{eq:Higgsing}
    \end{align}
\end{proposition}
\begin{proof}
    Since $K_\alpha = q_1^{\ell_\alpha}$, the length of the partition $\lambda_\alpha$ is restricted by Proposition~\ref{prop:Higgs1}, and thus the range of the index $i$ becomes $\{ 1,\ldots,\ell_\alpha \}$ for $\alpha = 1,\ldots,n$.
    Then, denoting $\ell_{\beta\alpha} = \ell_\beta - \ell_\alpha$, we have
    \begin{align}
    Z^\text{5d}_{\underline{\lambda}}
    & = \prod_{(\alpha,i) \neq (\beta, j)} \frac{(b_{\beta\alpha} q_2^{\lambda_{\alpha,i} - \lambda_{\beta,j}} q_1^{i-j}; q_2^{-1})_\infty}{(b_{\beta\alpha} q_2^{\lambda_{\alpha,i} - \lambda_{\beta,j}} q_1^{i-j+1}; q_2^{-1})_\infty} 
    \prod_{\substack{\alpha, f=1,\ldots,n \\ i = 1,\ldots,\ell_\alpha}} \frac{(\frac{\widetilde{m}_f}{b_\alpha} q_2^{\lambda_{\alpha,i}} q_1^{i};q_2^{-1})_\infty}{(\frac{m_f}{b_\alpha} q_2^{\lambda_{\alpha,i}} q_1^{i};q_2^{-1})_\infty}
    \nonumber \\
    & = \prod_{(\alpha,i) \neq (\beta, j)} \frac{(b_{\beta\alpha} q_2^{\lambda_{\alpha,\ell_\alpha+1-i} - \lambda_{\beta,\ell_\beta+1-j}} q_1^{\ell_{\beta\alpha} + j-i}; q_2^{-1})_\infty}{(b_{\beta\alpha} q_2^{\lambda_{\alpha,\ell_\alpha+1-i} - \lambda_{\beta,\ell_\beta+1-j}} q_1^{\ell_{\beta\alpha} +j-i+1}; q_2^{-1})_\infty} 
    \prod_{\substack{\alpha, f=1,\ldots,n \\ i = 1,\ldots,\ell_\alpha}} \frac{(\frac{\widetilde{m}_f}{b_\alpha} q_2^{\lambda_{\alpha,\ell_\alpha+1-i}} q_1^{\ell_\alpha+1-i};q_2^{-1})_\infty}{(\frac{m_f}{b_\alpha} q_2^{\lambda_{\alpha,\ell_\alpha+1-i}} q_1^{\ell_\alpha+1-i};q_2^{-1})_\infty}
    \, ,
    \end{align}
    which agrees with the vertex function under the identification~\eqref{eq:Higgsing}.
\end{proof}
This phenomenon is known as the triality between the $q$-deformed W-algebra conformal block and 3d/5d gauge theory partition functions~\cite{Aganagic:2013tta,Aganagic:2014oia,Aganagic:2015cta}.
As mentioned in \S\ref{sec:double_screening_A1}, the 3d/5d correspondence would be interpreted as a two-dimensional analog of the one-leg DT/PT correspondence.

\subsubsection{Elliptic conformal block}

We consider an elliptic analog of the conformal block discussed in Proposition~\ref{prop:conf_block_A1},
\begin{align}
        \Psi_{\mathscr{C}_{1\ldots n}}^{\underline{k}}(\underline{a};\underline{K}) & = \tr \left[ p^{d} \mathsf{Z}^{(k_1)}_{K_1}(a_1) \mathsf{Z}^{(k_2)}_{K_2}(a_2) \cdots \mathsf{Z}^{(k_n)}_{K_n}(a_n) \right] / \tr \left[ p^{d} \mathsf{Z}_{K_1}(a_1) \mathsf{Z}_{K_2}(a_2) \cdots \mathsf{Z}_{K_n}(a_n) \right]
        \nonumber \\ 
        & = \frac{\mathsf{c}_{\underline{k}}(\underline{K})}{k!} \oint_\gamma [\dd{\underline{x}}] \, \Phi_{2}(\underline{x},\underline{a}) \operatorname{Stab}_{\underline{k}}
        (\underline{x},\underline{a})
        \, ,
\end{align}
with the elliptic nome $p \in \mathbb{C}^\times$ and the principal part,
\begin{align}
    \Phi_{\sigma}(\underline{x},\underline{a}) = \prod_{i \neq j}^k \frac{\Gamma(q_{\bar{\sigma}} x_j/x_i;q_\sigma^{-1},p)}{\Gamma(x_j/x_i;q_\sigma^{-1},p)} \prod_{\substack{i=1,\ldots,k \\ \alpha=1,\ldots,n}} \frac{\Gamma( a_\alpha/x_i;q_\sigma^{-1},p)}{\Gamma(K_\alpha a_\alpha/x_i;q_\sigma^{-1},p)} 
    \, ,
\end{align}
where we denote by $\Gamma(\cdot)$ the elliptic gamma function~\eqref{eq:ell_gamma_fn}, and the trace is taken over the Fock space generated by all the free fields and $d$ is the corresponding degree counting operator. 
See, e.g., \cite{Kimura:2016dys}.
This principal part is reduced to the previous one presented in \eqref{eq:principal_A1} in the limit $p \to 0$.
This elliptic conformal block is identified with the partition function of 4d mass-deformed $\mathcal{N} = 2$ gauge theory (softly broken to $\mathcal{N}=1$ theory) on $\mathbb{C}_{q_2} \times \mathcal{E}_p$~\cite{Longhi:2019hdh}.

\paragraph{Elliptic vertex function}

Let us focus on the simplest example, which is an elliptic analog of the conformal block~\eqref{eq:TP1_block} associated with quasimaps to $T^\vee\mathbb{P}^1$.
We write $\Gamma(\cdot) = \Gamma(\cdot;q_2^{-1},p)$ and $\theta(\cdot) = \theta(\cdot;q_2^{-1})$.
Then, in the chamber $\mathscr{C}_{12}: a_1 \succ a_2$, we have%
\footnote{%
We can similarly obtain the vertex function for general $(n,k)$ given as a summation over the reverse partitions as discussed in \S\ref{sec:vort_part_fn}.
}
\begin{subequations}\label{eq:TP1_block12_elliptic}
\begin{align}
    \Psi_{\mathscr{C}_{12}}^{(1,0)}(a_{1,2};K_{1,2}) & = \oint_\gamma \frac{\dd{x}}{2 \pi \ii x} \frac{\Gamma(a_1/x,a_2/x)}{\Gamma(K_1 a_1/x,K_1 a_1/x)} \frac{\theta(z K_1 a_1/x) \theta(K_1)}{\theta(K_1 a_1/x) \theta(z K_1)}
    \, , \\ 
    \Psi_{\mathscr{C}_{12}}^{(0,1)}(a_{1,2};K_{1,2}) & = \oint_\gamma \frac{\dd{x}}{2 \pi \ii x} \frac{\Gamma(a_1/x,a_2/x)}{\Gamma(K_1 a_1/x,K_2 a_2/x)} \frac{\theta(z K_{12} a_2/x) \theta(K_2)}{\theta(K_2 a_2/x) \theta(z K_{12})} \frac{\theta(a_1/x)}{\theta(K_1 a_1 / x)}
    \, .
\end{align}
\end{subequations}
Similarly, in the other chamber $\mathscr{C}_{21}: a_1 \prec a_2$, we have
\begin{subequations}\label{eq:TP1_block21_elliptic}
\begin{align}
    \Psi_{\mathscr{C}_{21}}^{(1,0)}(a_{1,2};K_{1,2}) & = \oint_\gamma \frac{\dd{x}}{2 \pi \ii x} \frac{\Gamma(a_1/x,a_2/x)}{\Gamma(K_1 a_1/x,K_1 a_1/x)} \frac{\theta(z K_{12} a_1/x) \theta(K_1)}{\theta(K_1 a_1/x) \theta(z K_{12})} \frac{\theta(a_2/x)}{\theta(K_2 a_2 / x)}
    \, , \\ 
    \Psi_{\mathscr{C}_{21}}^{(0,1)}(a_{1,2};K_{1,2}) & = \oint_\gamma \frac{\dd{x}}{2 \pi \ii x} \frac{\Gamma(a_1/x,a_2/x)}{\Gamma(K_1 a_1/x,K_1 a_1/x)} \frac{\theta(z K_2 a_2/x) \theta(K_2)}{\theta(K_2 a_2/x) \theta(z K_2)}
    \, .
\end{align}
\end{subequations}

\begin{lemma}
    Taking the contour integral, we have the elliptic vertex function associated with each vacuum,
    \begin{subequations}
    \begin{align}
        \Psi_{\mathscr{C}_{12}}^{(1,0)}(a_{1,2};K_{1,2})\Big|_{\{a_1 q_2^{-d}\}} & = 
        \frac{\Gamma(q_2^{-1},a_{21})}{\Gamma(K_1,K_2 a_{21})} {_2 E_1}
        \left[
        \begin{matrix}
            q_2^{-1} K_1 & q_2^{-1} K_2 a_{21} \\ q_2^{-1} a_{21}
        \end{matrix} ; q_2^{-1}, p ; z K_{12} 
        \right]
        \, , \\
        \Psi_{\mathscr{C}_{12}}^{(1,0)}(a_{1,2};K_{1,2})\Big|_{\{a_2 q_2^{-d}\}} & = 
        \frac{\theta(z K_1 a_{12}, K_1) }{\theta(K_1 a_{12}, z K_1)}
        \frac{\Gamma(q_2^{-1},a_{12})}{\Gamma(K_2,K_1 a_{12})} {_2 E_1}
        \left[
        \begin{matrix}
            q_2^{-1} K_2 & q_2^{-1} K_1 a_{12} \\ q_2^{-1} a_{12}
        \end{matrix} ; q_2^{-1}, p ; z K_{12} 
        \right]
        \, , \\
        \Psi_{\mathscr{C}_{12}}^{(0,1)}(a_{1,2};K_{1,2})\Big|_{\{a_1 q_2^{-d}\}} & = 0
        \, , \\
        \Psi_{\mathscr{C}_{12}}^{(0,1)}(a_{1,2};K_{1,2})\Big|_{\{a_2 q_2^{-d}\}} & = 
        \frac{\theta(a_{12}) }{\theta(K_1 a_{12}) }
        \frac{\Gamma(q_2^{-1},a_{12})}{\Gamma(K_2,K_1 a_{12})} {_2 E_1}
        \left[
        \begin{matrix}
            q_2^{-1} K_2 & q_2^{-1} K_1 a_{12} \\ q_2^{-1} a_{12}
        \end{matrix} ; q_2^{-1}, p ; z K_{12} 
        \right]
        \, ,
    \end{align}
    \end{subequations}
    where ${}_2 E_1$ is the elliptic hypergeometric series (Definition~\ref{def:E-hypergeometric}).
    These contributions are summarized in a matrix form as follows,
    \begin{align}
        \Psi_{\mathscr{C}_{12}}(a_{1,2};K_{1,2}) & = 
        \begin{pmatrix}
            \Psi_{\mathscr{C}_{12}}^{(1,0)}(a_{1,2};K_{1,2}) \\ \Psi_{\mathscr{C}_{12}}^{(0,1)}(a_{1,2};K_{1,2})
        \end{pmatrix}
        \nonumber \\
        & = 
        \begin{pmatrix}
            1 & \displaystyle \frac{\theta(z K_1 a_{12},K_1)}{\theta(K_1 a_{12},z K_1) } \\ 
            0 & \displaystyle \frac{\theta(a_{12}) }{\theta(K_1 a_{12}) }
        \end{pmatrix}
        \begin{pmatrix}
        \displaystyle 
        \frac{\Gamma(q_2^{-1},a_{21})}{\Gamma(K_1,K_2 a_{21})} {_2 E_1}
        \left[
        \begin{matrix}
            q_2^{-1} K_1 & q_2^{-1} K_2 a_{21} \\ q_2^{-1} a_{21}
        \end{matrix} ; q_2^{-1}, p ; z K_{12} 
        \right] \\[1.5em]
        \displaystyle 
        \frac{\Gamma(q_2^{-1},a_{12})}{\Gamma(K_2,K_1 a_{12})} {_2 E_1}
        \left[
        \begin{matrix}
            q_2^{-1} K_2 & q_2^{-1} K_1 a_{12} \\ q_2^{-1} a_{12}
        \end{matrix} ; q_2^{-1}, p ; z K_{12} 
        \right]
        \end{pmatrix}
        \, .
    \end{align}
    Similarly, we have
    \begin{align}
        \Psi_{\mathscr{C}_{21}}(a_{1,2};K_{1,2}) & = 
        \begin{pmatrix}
            \Psi_{\mathscr{C}_{21}}^{(1,0)}(a_{1,2};K_{1,2}) \\ \Psi_{\mathscr{C}_{21}}^{(0,1)}(a_{1,2};K_{1,2})
        \end{pmatrix}
        \nonumber \\
        & = 
        \begin{pmatrix}
            \displaystyle \frac{\theta(a_{21})}{\theta(K_2 a_{21}) } & 0 \\
            \displaystyle \frac{\theta(z K_2 a_{21},K_2)}{\theta(K_2 a_{21},z K_2) } & 1
        \end{pmatrix}
        \begin{pmatrix}
        \displaystyle \frac{\Gamma(q_2^{-1},a_{21})}{\Gamma(K_1,K_2 a_{21})} {_2 E_1}
        \left[
        \begin{matrix}
            q_2^{-1} K_1 & q_2^{-1} K_2 a_{21} \\ q_2^{-1} a_{21}
        \end{matrix} ; q_2^{-1}, p ; z K_{12}
        \right] \\[1.5em]
        \displaystyle \frac{\Gamma(q_2^{-1},a_{12})}{\Gamma(K_2,K_1 a_{12})} {_2 E_1}
        \left[
        \begin{matrix}
            q_2^{-1} K_2 & q_2^{-1} K_1 a_{12} \\ q_2^{-1} a_{12}
        \end{matrix} ; q_2^{-1}, p ; z K_{12}
        \right]
        \end{pmatrix}
        \, ,
    \end{align}
    from which we obtain the elliptic $R$-matrix as before (Proposition~\ref{prop:R-matrix}).
\end{lemma}

\paragraph{Elliptic $q$-KZ equation}

Let us consider the difference equation for the two-point elliptic conformal block.
\begin{lemma}
    We have
    \begin{align}
        \Psi^{(1,0)}_{\mathscr{C}_{12}}(a_1,pa_2;K_{1,2}) = \Psi^{(1,0)}_{\mathscr{C}_{21}}(a_1,a_2;K_{1,2})\Big|_{z \to z K_2^{-1}}
        \, , \quad 
        \Psi^{(0,1)}_{\mathscr{C}_{12}}(a_1,pa_2;K_{1,2}) = \Psi^{(0,1)}_{\mathscr{C}_{21}}(a_1,a_2;K_{1,2})\Big|_{z \to z K_1}
        \, .
    \end{align}
\end{lemma}
\begin{proof}
    This immediately follows from the integral formulas \eqref{eq:TP1_block12_elliptic} and \eqref{eq:TP1_block21_elliptic}.
\end{proof}
Defining the weight factor (see also \eqref{eq:Zz_weight}),
\begin{align}
    \kappa_\alpha = K_\alpha q_1^{-2k_\alpha} \, ,
    \label{eq:weight_fac}
\end{align}
and imposing the zero weight condition, $K_{12} = q_1^2$, we may rewrite 
\begin{align}
        \Psi^{(1,0)}_{\mathscr{C}_{12}}(a_1,pa_2;K_{1,2}) = \Psi^{(1,0)}_{\mathscr{C}_{21}}(a_1,a_2;K_{1,2})\Big|_{z \to z \kappa_2^{-1}}
        \, , \quad 
        \Psi^{(0,1)}_{\mathscr{C}_{12}}(a_1,pa_2;K_{1,2}) = \Psi^{(0,1)}_{\mathscr{C}_{21}}(a_1,a_2;K_{1,2})\Big|_{z \to z \kappa_2^{-1}}
        \, .
\end{align}
Hence, we have
\begin{align}
        \Psi_{\mathscr{C}_{12}}\Big|_{a_2 \to p a_2} = \Psi_{\mathscr{C}_{21}}\Big|_{z \to z \kappa_2^{-1}} \, .
\end{align}
We remark that the role of the zero weight condition to obtain the difference equation was first discussed in~\cite{Foda:1993fg}.
See also a recent monograph on this subject~\cite{Konno2020}.

\begin{proposition}
    Under the zero weight condition $K_{12} = q_1^2$, we have the elliptic $q$-KZ equation, 
    \begin{align}
        \Psi_{\mathscr{C}_{12}}\Big|_{\substack{z \to z \kappa_2 \\ a_2 \to p a_2}} = \Psi_{\mathscr{C}_{21}} = R_{12 \to 21} \Psi_{\mathscr{C}_{12}} 
        \, .
    \end{align}
    Similarly, we have 
    \begin{align}
        \Psi_{\mathscr{C}_{12}}\Big|_{\substack{z \to z \kappa_1 \\ a_1 \to p a_1}}
        = R_{21 \to 12} \Psi_{\mathscr{C}_{21}}\Big|_{\substack{z \to z \kappa_1 \\ a_1 \to p a_1}} = R_{21 \to 12}\Big|_{\substack{z \to z \kappa_1 \\ a_1 \to p a_1}} \Psi_{\mathscr{C}_{12}}
        \, .
    \end{align}
\end{proposition}
This is the simplest case of the elliptic $q$-KZ equation.
Let us discuss a general situation as follows.


\begin{lemma}\label{lemma:n_gen}
Let $\Psi_{\mathscr{C}_{1\ldots n}}$ be the elliptic conformal block in the chamber $\mathscr{C}_{1\ldots n}$.
Then, under the zero weight condition $K_1 \cdots K_n = q_1^{2k}$, the following holds,
\begin{align}
    \Psi_{\mathscr{C}_{1\ldots n}}^{\underline{k}}\Big|_{a_n \to p a_n} = \Psi_{\mathscr{C}_{n1\ldots n-1}}^{\underline{k}}\Big|_{z \to z \kappa_n^{-1}}
    \, ,
\end{align}
where $\kappa_n$ is the weight factor defined in \eqref{eq:weight_fac}.
\end{lemma}
\begin{proof}
Denoting $\Gamma(\cdot) = \Gamma(\cdot;q_2^{-1},p)$ and $\theta(\cdot) = \theta(\cdot,q_2^{-1})$, the elliptic conformal block associated with the configuration $\underline{k}=(k_1,\ldots,k_n)$ is given by
\begin{align}
    \Psi_{\mathscr{C}_{1\ldots n}}^{\underline{k}}
    & = \frac{1}{k!} \oint_\gamma [\dd{\underline{x}}] \prod_{i \neq j}^k \frac{\Gamma(q_1 x_{ij})}{\Gamma(x_{ij})} \prod_{\substack{i=1,\ldots,k \\ \alpha=1,\ldots,n}} \frac{\Gamma( a_\alpha/x_i)}{\Gamma(K_\alpha a_\alpha/x_i)} 
    \nonumber \\
    & \qquad \times
    \operatorname*{Sym}_{\underline{x}}
    \prod_{i<j}^k \frac{\theta(q_1 x_{ji})}{\theta(x_{ji})}
    \prod_{i=1}^k \left[ \left( \prod_{\alpha < c(i)} \frac{\theta(a_\alpha/x_i)}{\theta(K_\alpha a_\alpha / x_i)} \right) \frac{\theta(z q_1^{-2\rho_i} K_{i}^c a_{c(i)}/x_i) \theta(K_{c(i)})}{\theta(K_{c(i)} a_{c(i)} / x_i) \theta(z q_1^{-2\rho_i} K_i^c)} \right]
    \, .
\end{align}
Then, recalling $K_i^c = K_1 K_2 \cdots K_{c(i)} = K_{1\ldots c(i)}$, we have
\begin{align}
    \Psi_{\mathscr{C}_{1\ldots n}}^{\underline{k}}\Big|_{a_n \to p a_n} 
    & = \frac{1}{k!} \oint_\gamma [\dd{\underline{x}}] \prod_{i \neq j}^k \frac{\Gamma(q_1 x_{ij})}{\Gamma(x_{ij})} \prod_{\substack{i=1,\ldots,k \\ \alpha=1,\ldots,n}} \frac{\Gamma( a_\alpha/x_i)}{\Gamma(K_\alpha a_\alpha/x_i)} \operatorname*{Sym}_{\underline{x}} \prod_{i<j}^k \frac{\theta(q_1 x_{ji})}{\theta(x_{ji})}
    \nonumber \\
    & \qquad \times
    \prod_{i=k_n+1}^k \left[ \frac{\theta(a_n/x_i)}{\theta(K_n a_n/x_i)} \left( \prod_{\alpha < c(i)} \frac{\theta(a_\alpha/x_i)}{\theta(K_\alpha a_\alpha / x_i)} \right) \frac{\theta(z q_1^{-2\rho_i} K_{1\ldots c(i)} a_{c(i)}/x_i) \theta(K_{c(i)})}{\theta(K_{c(i)} a_{c(i)} / x_i) \theta(z q_1^{-2\rho_i} K_{1\ldots c(i)})} \right]
    \nonumber \\
    & \qquad \times \prod_{i=1}^{k_n} \left[ \left( \prod_{\alpha = 1}^n \frac{\theta(a_\alpha/x_i)}{\theta(K_\alpha a_\alpha / x_i)} \right) \frac{\theta(p z q_1^{-2\rho_i} K_{1\ldots n} a_{n}/x_i) \theta(K_{n})}{\theta(p K_{n} a_{n} / x_i) \theta(z q_1^{-2\rho_i} K_{1\ldots n})} \right]
    \nonumber \\
    & \stackrel{\substack{x_i \to p x_i \\ i=1,\ldots,k_n}}{=} \frac{1}{k!} \oint_\gamma [\dd{\underline{x}}] \prod_{i \neq j}^k \frac{\Gamma(q_1 x_{ij})}{\Gamma(x_{ij})} \prod_{\substack{i=1,\ldots,k \\ \alpha=1,\ldots,n}} \frac{\Gamma( a_\alpha/x_i)}{\Gamma(K_\alpha a_\alpha/x_i)} \operatorname*{Sym}_{\underline{x}} \prod_{1 \le i<j \le k}' \frac{\theta(q_1 x_{ji})}{\theta(x_{ji})}
    \nonumber \\
    & \qquad \times
    \prod_{i=k_n+1}^k \left[ \frac{\theta(a_n/x_i)}{\theta(K_n a_n/x_i)} \left( \prod_{\alpha < c(i)} \frac{\theta(a_\alpha/x_i)}{\theta(K_\alpha a_\alpha / x_i)} \right) \frac{\theta(z q_1^{-2\rho_i} K_{1\ldots c(i)} a_{c(i)}/x_i) \theta(K_{c(i)})}{\theta(K_{c(i)} a_{c(i)} / x_i) \theta(z q_1^{-2\rho_i} K_{1\ldots c(i)})} \right]
    \nonumber \\
    & \qquad \times 
    \prod_{i=1}^{k_n}
    \left[ \frac{\theta(z q_1^{-2\rho_i} K_{1\ldots n} a_{n}/x_i) \theta(K_{n})}{\theta(K_{n} a_{n} / x_i) \theta(z q_1^{-2\rho_i} K_{1\ldots n})} \right]
    \, ,
\end{align}
where we write
\begin{align}
    \prod_{1 \le i<j \le k}' \frac{\theta(q_1 x_{ji})}{\theta(x_{ji})} = \prod_{1 \le i \le k_n < j \le k} \frac{\theta(q_1 x_{ij})}{\theta(x_{ij})} \prod_{1 \le i<j \le k_n} \frac{\theta(q_1 x_{ji})}{\theta(x_{ji})} \prod_{k_n + 1 \le i<j \le k} \frac{\theta(q_1 x_{ji})}{\theta(x_{ji})}
    \, .
\end{align}
Moreover, due to the symmetrization, we may change the ordering of the $x$-variables, $(x_1, \ldots, x_k) \mapsto (x_{k_n+1},\ldots,x_k,x_1,\ldots,x_{k_n})$, to compute
\begin{align}
    \Psi_{\mathscr{C}_{1\ldots n}}^{\underline{k}}\Big|_{a_n \to p a_n} 
    & = \frac{1}{k!} \oint_\gamma [\dd{\underline{x}}] \prod_{i \neq j}^k \frac{\Gamma(q_1 x_{ij})}{\Gamma(x_{ij})} \prod_{\substack{i=1,\ldots,k \\ \alpha=1,\ldots,n}} \frac{\Gamma( a_\alpha/x_i)}{\Gamma(K_\alpha a_\alpha/x_i)} \operatorname*{Sym}_{\underline{x}} \prod_{1 \le i<j \le k} \frac{\theta(q_1 x_{ji})}{\theta(x_{ji})}
    \nonumber \\
    & \qquad \times
    \prod_{i=1}^{k-k_n} \left[ \frac{\theta(a_n/x_i)}{\theta(K_n a_n/x_i)} \left( \prod_{\alpha < c(i)} \frac{\theta(a_\alpha/x_i)}{\theta(K_\alpha a_\alpha / x_i)} \right) \frac{\theta(z q_1^{-2\rho_{i+k_n}} K_{1\ldots c(i)} a_{c(i)}/x_i) \theta(K_{c(i)})}{\theta(K_{c(i)} a_{c(i)} / x_i) \theta(z q_1^{i+k_n-k} K_{1\ldots c(i)})} \right]
    \nonumber \\
    & \qquad \times 
    \prod_{i=k-k_n+1}^{k}
    \left[ \frac{\theta(z q_1^{-2\rho_{i+k_n-k}} K_{1\ldots n} a_{n}/x_i) \theta(K_{n})}{\theta(K_{n} a_{n} / x_i) \theta(z q_1^{i+k_n-2k} K_{1\ldots n})} \right]
    \, .
\end{align}
This is identified with $\Psi_{\mathscr{C}_{n1\ldots n-1}}^{\underline{k}}\Big|_{z \to z K_n^{-1} q_1^{k_n}}$ under the zero weight condition $K_{1\ldots n} = q_1^{2k}$.
\end{proof}

Set
\begin{align}
    \kappa_\alpha^c = \prod_{\alpha' < \alpha} \kappa_{\alpha'}
    \, .
\end{align}
Then, we have the dynamical elliptic $q$-KZ equation for the elliptic conformal block.    
\begin{theorem}\label{thm:e_qKZ}
    The elliptic conformal block obeys the dynamical elliptic $q$-KZ equation,
    \begin{align}
        \Psi_{\mathscr{C}_{1\ldots n}}\Big|_{\substack{z \to z \kappa_\alpha \\ a_\alpha \to p a_\alpha}} = R_{\alpha+1,\alpha}(z \kappa^c_{\alpha+1}) \cdots R_{n,\alpha} (z \kappa^c_{n})\Big|_{a_\alpha \to p a_\alpha} R_{1,\alpha} (z) \cdots R_{\alpha-1,\alpha} (z \kappa^c_{\alpha-1}) \Psi_{\mathscr{C}_{1\ldots n}}
        \, .
    \end{align}
    The zero weight condition implies the periodicity $\kappa_{\alpha}^c = \kappa_{\alpha+n}^c$ for $\alpha = 1,\ldots,n$.
\end{theorem}
This elliptic $q$-KZ solution was first constructed in the context of elliptic quantum group $U_{q,p}(\widehat{\mathfrak{sl}}_2)$~\cite{Foda:1993fg,Felder:1994pb,Konno:2017mos}.
We emphasize that the conformal block of the W-algebra $\mathrm{W}_{q_{1,2}}(A_1)$ itself satisfies the $q$-KZ equation, directly establishing the quantum $q$-Langlands correspondence between the electric block and the magnetic block.
Moreover, we do not need to insert an additional stable envelope for the elliptic setup compared with the trigonometric situation~\cite{Aganagic:2016jmx}.

\subsubsection{Trigonometric conformal block}

We denote the plethystic exponential by $\operatorname{PE}$.
We write $(\cdot)_\infty = (\cdot;q_2^{-1})_\infty$ and $[\cdot] = \operatorname{PE}[\cdot]^{-1}$.
We consider the trigonometric conformal block associated with the configuration $\underline{k}=(k_1,\ldots,k_n)$ given by
\begin{align}
    \Psi_{\mathscr{C}_{1\ldots n}}^{\underline{k}}
    & = \frac{1}{k!} \exp\left( \frac{\log z}{\log q_2^{-1}} \sum_{\alpha=1}^n \frac{\mu_\alpha}{2} \log a_\alpha \right) \oint_\gamma [\dd{\underline{x}}] \prod_{i=1}^k x_i^{-\frac{\log z}{\log q_2^{-1}}} \prod_{i \neq j}^k \frac{(x_{ij})_\infty}{(q_1 x_{ij})_\infty} \prod_{\substack{i=1,\ldots,k \\ \alpha=1,\ldots,n}} \frac{(K_\alpha a_\alpha/x_i)_\infty}{( a_\alpha/x_i)_\infty} 
    \nonumber \\
    & \qquad \times
    \operatorname*{Sym}_{\underline{x}} \left[
    \prod_{i<j}^k \frac{[q_1 x_{ji}]}{[x_{ji}]} 
    \prod_{i=1}^k \prod_{\alpha < c(i)} \frac{[a_\alpha/x_i]}{[K_\alpha a_\alpha / x_i]} \right]
    \, .
\end{align}

\begin{lemma}
    The trigonometric version of Lemma~\ref{lemma:n_gen} holds,
\begin{align}
    \Psi^{\underline{k}}_{\mathscr{C}_{1\ldots n}}\Big|_{a_n \to q_2^{-1} a_n} = z^{\frac{\mu_n}{2}-k_n} \Psi^{\underline{k}}_{\mathscr{C}_{n1\ldots n-1}}
    \, ,
\end{align}    
    where the factor $\frac{\mu_n}{2}-k_n$ corresponds to the weight associated with the $n$-th operator (see also~\eqref{eq:weight_fac}).
\end{lemma}
\begin{proof}
We have
\begin{align}
    & \Psi_{\mathscr{C}_{1\ldots n}}^{\underline{k}}\Big|_{a_n \to q_2^{-1} a_n}
    \nonumber \\
    & = \frac{z^{\frac{\mu_\alpha}{2}}}{k!} \exp\left( \frac{\log z}{\log q_2^{-1}} \sum_{\alpha=1}^n \frac{\mu_\alpha}{2} \log a_\alpha \right) \oint_\gamma \prod_{i=1}^k \frac{\dd{x}_i}{2 \pi \ii x_i} x_i^{-\frac{\log z}{\log q_2^{-1}}} \prod_{i \neq j}^k \frac{(x_{ij})_\infty}{(q_1 x_{ij})_\infty} \prod_{\substack{i=1,\ldots,k \\ \alpha=1,\ldots,n}} \frac{(K_\alpha a_\alpha/x_i)_\infty}{( a_\alpha/x_i)_\infty} 
    \nonumber \\
    & \qquad \times
    \operatorname*{Sym}_{\underline{x}}
    \prod_{i<j}^k \frac{[q_1 x_{ji}]}{[x_{ji}]} 
    \prod_{i=1}^k \left[ \frac{[a_n/x_i]}{[K_n a_n / x_i]} \prod_{\alpha < c(i)} \frac{[a_\alpha/x_i]}{[K_\alpha a_\alpha / x_i]} \right]
    \nonumber \\
    & \stackrel{\substack{x_i \to q_2^{-1} x_i \\ i=1,\ldots,k_n}}{=} \frac{z^{\frac{\mu_\alpha}{2}-k_n}}{k!} \exp\left( \frac{\log z}{\log q_2^{-1}} \sum_{\alpha=1}^n \frac{\mu_\alpha}{2} \log a_\alpha \right) \oint_\gamma \prod_{i=1}^k \frac{\dd{x}_i}{2 \pi \ii x_i} x_i^{-\frac{\log z}{\log q_2^{-1}}} \prod_{i \neq j}^k \frac{(x_{ij})_\infty}{(q_1 x_{ij})_\infty} \prod_{\substack{i=1,\ldots,k \\ \alpha=1,\ldots,n}} \frac{(K_\alpha a_\alpha/x_i)_\infty}{( a_\alpha/x_i)_\infty}
    \nonumber \\
    & \hspace{10em} \times 
    \operatorname*{Sym}_{\underline{x}}
    \prod_{1 \le i<j \le k}' \frac{[q_1 x_{ji}]}{[x_{ji}]} 
    \prod_{i=k_n+1}^k \left[ \frac{[a_n/x_i]}{[K_n a_n / x_i]} \prod_{\alpha < c(i)} \frac{[a_\alpha/x_i]}{[K_\alpha a_\alpha / x_i]} \right]
    \, ,
\end{align}
where we write
\begin{align}
    \prod_{1 \le i<j \le k}' \frac{[q_1 x_{ji}]}{[x_{ji}]} = \prod_{1 \le i \ge k_n < j \ge k} \frac{[q_1 x_{ij}]}{[x_{ij}]} \prod_{1 \le i < j \le k_n} \frac{[q_1 x_{ji}]}{[x_{ji}]} \prod_{k_n+1 \le i < j \le k} \frac{[q_1 x_{ji}]}{[x_{ji}]}
    \, .
\end{align}
Changing the ordering of the $x$-variables, $(x_1,\ldots,x_k) \to (x_{k_n+1},\ldots,x_k,x_1,\ldots,x_{k_n})$ (the stable envelope is invariant under this change due to the symmetrization), we obtain the result.
\end{proof}

\begin{proposition}\label{prop:t_qKZ}
    Let $h_\alpha$ be an operator that counts the weight of the $\alpha$-th operator.
    Then, we have the $q$-KZ equation,
    \begin{align}
        \Psi_{\mathscr{C}_{1\ldots n}}\Big|_{a_\alpha \to q_2^{-1} a_\alpha} = R_{\alpha+1,\alpha} (q_2 a_{\alpha+1} / a_\alpha) \cdots R_{n,\alpha} (q_2 a_n / a_\alpha) z^{h_n} R_{1,\alpha} (a_1/a_\alpha) \cdots R_{\alpha-1,\alpha}(a_{\alpha - 1}/a_\alpha) \Psi_{\mathscr{C}_{1\ldots n}}
        \, .
    \end{align}
\end{proposition}
Compared with the elliptic case (Theorem~\ref{thm:e_qKZ}), there is no dependence of the dynamical parameter $z$ in the trigonometric $q$-KZ equation except for the factor $z^{h_n}$.

\subsection{Double screening current}\label{sec:double_screening_A1}

We now consider the vertex operator screened by both types of the screening currents $\mathsf{S}_{1,2}$.
We assume $|q_{1,2}|>1$.
\begin{definition}
    We define the $(m_1|m_2)$-screened vertex operator,
    \begin{align}
        \mathsf{Z}_K^{(m_1|m_2)}(a) = \oint_\gamma 
        [\dd{\underline{x}}] \prod_{\sigma \in \underline{\mathbf{2}}}
        \left( \prod_{i=1,\ldots,m_\sigma}^{\curvearrowleft} \varphi_{\sigma,K}(a/x_{\sigma,i},zK) \mathsf{S}_\sigma(x_{\sigma,i}) \right) \mathsf{Z}_K(a)
        \, , \qquad 
        [\dd{\underline{x}}] = \prod_{\sigma \in \underline{\mathbf{2}}} \prod_{i=1}^{m_\sigma} \frac{\dd{x}_{\sigma,i}}{2 \pi \ii x_{\sigma,i}} 
        \, .
    \end{align}
    We apply the ordered product under the radial ordering configuration,
    \begin{align}
        \begin{cases}
            x_{1,m_1} \succ x_{1,m_1-1} \succ \cdots \succ x_{1,2} \succ x_{1,1} \succ a \\
            x_{2,m_2} \succ x_{2,m_2-1} \succ \cdots \succ x_{2,2} \succ x_{2,1} \succ a
        \end{cases}
    \end{align}
    Since $\mathsf{S}_1$ and $\mathsf{S}_2$ commute, we do not need to specify the ordering between the $x_1$- and $x_2$-variables.
\end{definition}

Let $(\underline{k}_1|\underline{k}_2) = (k_{1,1},k_{1,2},\ldots,k_{1,n} \mid k_{2,1},k_{2,2},\ldots,k_{2,n})$.
We write $\bar{\sigma} = \underline{\mathbf{2}} \backslash \sigma$ for $\sigma \in \underline{\mathbf{2}}$.
Then we have the conformal block of the screened vertex operators in the chamber denoted by $\mathscr{C}_{12\ldots n}: a_1 \succ a_2 \succ \cdots \succ a_n$,
\begin{align}
    & \Psi_{\mathscr{C}_{1\ldots n}}^{(\underline{k}_1|\underline{k}_2)}(\underline{a};\underline{K}) = \bra{0} \mathsf{Z}^{(k_{1,1}|k_{2,1})}_{K_1}(a_1) \mathsf{Z}^{(k_{1,2}|k_{2,2})}_{K_2}(a_2) \cdots \mathsf{Z}^{(k_{1,n}|k_{2,n})}_{K_n}(a_n) \ket{0} / \bra{0} \mathsf{Z}_{K_1}(a_1) \mathsf{Z}_{K_2}(a_2) \cdots \mathsf{Z}_{K_n}(a_n) \ket{0}
    \nonumber \\ 
    & = \prod_{\sigma \in \underline{\mathbf{2}}} \frac{\mathsf{c}_{\sigma,\underline{k}_\sigma}(\underline{K})}{k_\sigma!} \oint_\gamma [\dd{\underline{x}}] \prod_{\sigma \in \underline{\mathbf{2}}} \Phi_{\sigma}(\underline{x}_\sigma,\underline{a}) \operatorname{Stab}_{\sigma,{\underline{k}}_\sigma}(\underline{x}_\sigma,\underline{a}) 
    \prod_{\substack{i=1,\ldots,k_1 \\ j = 1,\ldots,k_2}}\frac{-q_{12} x_{1,i} x_{2,j}}{(1 - q_2 x_{1,i}/x_{2,j})(1 - q_1 x_{2,j}/x_{1,i})}
    \, .
\end{align}    
This conformal block computes the intersecting vortex theory supported on $\mathbb{C}_{q_1}$ and $\mathbb{C}_{q_2}$~\cite{Gomis:2016ljm,Pan:2016fbl,Kimura:2021ngu}.
From the geometric point of view, it is constructed from the intersection of (moduli spaces of) quasimaps, which gives rise to the doubled version of the handsaw quiver (Fig.~\ref{fig:double_A1}).
It can be seen as $\widehat{A}_1$ quiver theory as it involves two gauge nodes with two bifundamental matters.
In fact, this configuration is closely related to $\mathrm{U}(k_1|k_2)$ supergroup gauge theory~\cite{Kimura:2019msw,Kimura:2023iup}.

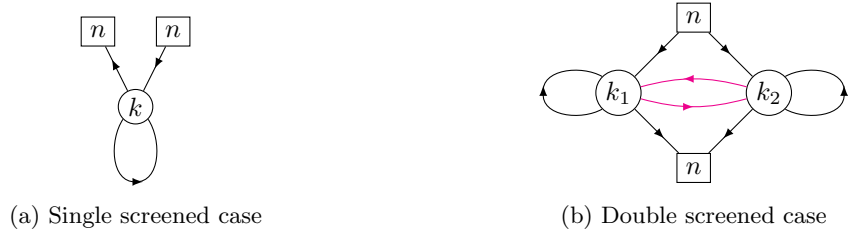
\begin{figure}[t]
    \centering
    \begin{subfigure}[b]{0.4\textwidth}
         \centering
         \begin{tikzpicture}
             \node[circle,draw=black,minimum size=5pt,inner sep=1.5pt] (k) at (0,0) {$k$};
             \node[rectangle,draw=black] (n1) at (.5,1) {$n$};
             \node[rectangle,draw=black] (n2) at (-.5,1) {$n$};
             \draw [->-] (n1) -- (k);
             \draw [-<-] (n2) -- (k);
             \draw (k) to [bend right = 45,in=260] ++(0,-1) to [bend right = 45,out=-80] (k);
             \draw [-latex] (0,-1) -- ++ (.1,0);
         \end{tikzpicture}
         \caption{Single screened case}
         \label{fig:single_A1}
     \end{subfigure}
     \qquad
    \begin{subfigure}[b]{0.4\textwidth}
         \centering
         \begin{tikzpicture}
             \node[circle,draw=black,minimum size=5pt,inner sep=1.5pt] (k1) at (-1,0) {$k_1$};
             \node[circle,draw=black,minimum size=5pt,inner sep=1.5pt] (k2) at (1,0) {$k_2$};
             \node[rectangle,draw=black] (n1) at (0,1) {$n$};
             \node[rectangle,draw=black] (n2) at (0,-1) {$n$};
             \draw [->-] (n1) -- (k1);
             \draw [-<-] (n2) -- (k1);
             \draw [->-] (n1) -- (k2);
             \draw [-<-] (n2) -- (k2);             
             \draw (k1) to [bend right = 45,in=260] ++(-1,0) to [bend right = 45,out=-80] (k1);
             \draw (k2) to [bend right = 45,in=260] ++(1,0) to [bend right = 45,out=-80] (k2);             
             \draw [-latex] (-2,0) -- ++ (0,.1);
             \draw [-latex] (2,0) -- ++ (0,.1);
             \draw [->-,magenta] (k1) to [bend right = 15] (k2);
             \draw [-<-,magenta] (k1) to [bend left = 15] (k2);             
         \end{tikzpicture}
         \caption{Double screened case}
         \label{fig:double_A1}
     \end{subfigure}     
    \caption{Quivers for $\mathrm{W}_{q_{1,2}}(A_1)$ block. The red edge describes the interaction of two nodes (the so-called 0d theory contribution).}
    \label{fig:quiver_A1}
\end{figure}

\paragraph{Pole structure}

We reparametrize the $x_{1,2}$-variables, $(x_{\sigma,i})_{i=1,\ldots,k_\sigma} = (x_{\sigma,\alpha,i})_{\alpha=1,\ldots,n,i=1,\ldots,k_{\sigma,\alpha}}$.
There are three possible types of the poles for $\sigma \in \underline{\mathbf{2}}$, $\bar{\sigma} = \underline{\mathbf{2}} \backslash \sigma$, (i) $x_{\sigma,\alpha,i} = a_\alpha q_{\sigma}^{\mathbb{Z}_{\le 0}}$, (ii) $x_{\sigma,\alpha,i} = x_{\sigma,\alpha,j (< i)} q_{\bar{\sigma}} q_\sigma^{\mathbb{Z}_{\le 0}}$, and (iii) $x_{\sigma,\alpha,i} = x_{\bar{\sigma},\alpha,j} q_\sigma$.
Here we take the following poles:
\begin{align}
    x_{\sigma,\alpha,1} = a_\alpha q_{\sigma}^{-\nu_{\sigma,\alpha,1}} 
    \, , \qquad &
    x_{\sigma,\alpha,i} = x_{\sigma,\alpha,i-1} q_{\bar\sigma} q_\sigma^{-\nu_{\sigma,\alpha,i} + \nu_{\sigma,\alpha,i-1}} = a_\alpha q_{\bar\sigma}^{i-1} q_\sigma^{-\nu_{\sigma,\alpha,i}}
    \, ,
\end{align}
where $(\nu_{\sigma,\alpha})_{\alpha = 1,\ldots,n}$ is a set of reverse partitions,
\begin{align}
    \nu_{1,\alpha} = (0 \le \nu_{1,\alpha,1} \le \cdots \le \nu_{1,\alpha,k_{1,\alpha}})
    \, , \qquad 
    \nu_{2,\alpha} = (-k_{1,\alpha} \le \nu_{2,\alpha,1} \le \cdots \le \nu_{2,\alpha,k_{2,\alpha}})
    \, .
\end{align}
The entries of $\nu_{1,\alpha}$ are all non-negative, while those of $\nu_{2,\alpha}$ can be negative, $\nu_{2,\alpha,i} \ge - k_{1,\alpha}$. 
We take the type (iii) poles for the $x_2$-variable to have the negative value through the following path, $x_{1,\alpha,1} = a_\alpha$, $x_{1,\alpha,2} = a_\alpha q_2$, \ldots, $x_{1,\alpha,d} = a_\alpha q_2^{d-1}$, then $x_{2,\alpha,1} = x_{1,\alpha,d} q_2 = a_\alpha q_2^{d}$, which corresponds to $\nu_{2,\alpha,1} = -d$, and also $\nu_{1,\alpha,i} = 0$ for $i = 1,\ldots, d$ with $d = \#\{ \nu_{1,\alpha,i} = 0, i = 1,\ldots,k_{1,\alpha} \} = k_{1,\alpha} - \ell(\nu_{1,\alpha})$ (the length of partition $\lambda$ denoted by $\ell(\lambda) = \lambda_1^\rmT$).
\begin{example}\label{ex:2dPT_n=1}
    In the case $n = 1$ with $k_1 = k_{1,1} = 6$, $k_2 = k_{2,1} = 8$, we may have the following configuration,
    
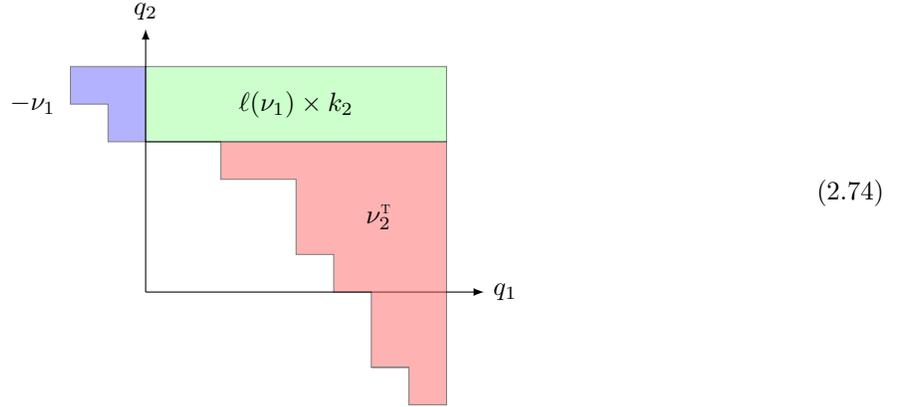
\begin{align}
    \begin{tikzpicture}[baseline=(current bounding  box.center),scale=.5]
        \draw[-latex] (0,0)--(0,7) node [above] {$q_2$};
        \draw[-latex] (0,0)--(9,0) node [right] {$q_1$};
        \filldraw[draw=black,fill=blue!60,opacity=.5] (0,4) -- ++(-1,0) -- ++ (0,1) -- ++ (-1,0) -- ++(0,1) -- ++ (2,0) -- cycle;
        \draw (0,4) -- ++(2,0);
        \filldraw[draw=black,fill=red!60,opacity=.5] (2,4) -- ++(0,-1) -- ++(2,0) -- ++(0,-2) -- ++(1,0) -- ++(0,-1) -- ++(1,0) -- ++ (0,-2) -- ++ (1,0) -- ++ (-1,0) -- ++(1,0) -- ++(0,-1) -- ++(1,0) -- ++(0,7) -- cycle;
        \filldraw[fill=green!40,opacity=.5] (0,4) -- ++(8,0) -- ++ (0,2) -- ++(-8,0) -- cycle;
        \node at (-3,5) {$-\nu_1$};
        \node at (6.2,2) {$\nu_2^\rmT$};
        \node at (4,5) {$\ell(\nu_1) \times k_2$};
    \end{tikzpicture}
    \end{align}
    where 
    \begin{align}
        \nu_1 = (0,0,0,0,1,2) 
        \, , \qquad 
        \nu_2 = (-4,-4,-3,-3,-1,0,2,3)
        \, ,
    \end{align}
    with $\ell(\nu_1) = 2$ and thus $d = 4$.
\end{example}

Combining $-\nu_{1,\alpha}$, $\nu_{2,\alpha}^\rmT$, and also the rectangular part of the size $\ell(\nu_{1,\alpha}) \times k_{2,\alpha}$ (the green part in Example~\ref{ex:2dPT_n=1}), the whole part is interpreted as a (reverse) fat hook partition, which parametrizes the irreducible representation of supergroup $\mathrm{U}(k_{1,\alpha}|k_{2,\alpha})$.
This counting rule is naturally seen as a two-dimensional version of the two-leg PT count, i.e., PT2 vertex for $\mathbb{C}^2$:
The equivariant fixed points are described by a two-dimensional partition with two boundary conditions parametrized by one-dimensional partitions ($k_{1,\alpha}$ and $k_{2,\alpha}$ for each $\alpha = 1, \ldots,n$).
See~\cite{Monavari:2025xrs} for a recent progress in this direction.

\begin{proposition}\label{prop:Higgs2}
Let $K_\alpha = q_1^{\mu_{2,\alpha}} q_2^{\mu_{1,\alpha}}$, $\mu_{\sigma,\alpha} \in \mathbb{Z}_{\ge 0}$ for $\sigma=1,2$, $\alpha = 1,\ldots,n$.
Then, $\Psi^{\underline{k}} = 0$ unless $0 \le k_{\sigma,\alpha} \le \mu_{\sigma,\alpha}$.
\end{proposition}
This is a double version of Proposition~\ref{prop:Higgs1}, which is also called the (two-dimensional version of) pit condition~\cite{Bershtein:2018SM}.
Moreover, Proposition~\ref{prop:Z_5d} can be also generalized to the current case:
Starting with the 5d $\mathcal{N}=1$ $\mathrm{U}(n)$ gauge theory on $\mathbb{C}_{q_1} \times \mathbb{C}_{q_2} \times \mathbb{S}^1$ with $n$ fundamental and $n$ antifundamental matters with the Higgsing condition $K_\alpha = m_\alpha/\widetilde{m}_\alpha = q_1^{\ell_{2,\alpha}} q_2^{\ell_{1,\alpha}}$, we obtain the intersecting defect partition function, which realizes the double screening current situation~\cite{Dimofte:2010tz,Bonelli:2011fq,Gomis:2016ljm,Pan:2016fbl,Gorsky:2017hro,Nieri:2017ntx,Kimura:2021ngu}.
As mentioned above, the double screening current situation may be seen as the two-dimensional two-leg PT count.
From this point of view, the 3d/5d relation is then interpreted as the two-dimensional version of DT/PT correspondence, equivalence of D0 brane counting and D2 brane counting.
We see the higher-dimensional uplift of this story in the following section.

\section{W$_{q_{1,2,3,4}}(\widehat{A}_0)$ conformal block}

We consider the W-algebra associated with $\widehat{A}_0$ quiver (Jordan quiver) introduced in~\cite{Kimura:2015rgi}, which depends on four parameters $\underline{q} = ( q_\sigma )_{\sigma \in \underline{\mathbf{4}}}$ with the condition,
\begin{align}
 q_1 q_2 q_3 q_4 = 1 \, .
 \label{eq:CY4_cond}
\end{align}
Identifying $\{q_{1,2,3,4}\}$ with the equivariant parameters associated with $\mathbb{C}^4$, this is the Calabi--Yau 4-fold condition.

\begin{definition}\label{def:OPEs_A0}
    For $A \subseteq \underline{\mathbf{4}}$, we write $\bar{A} = \underline{\mathbf{4}} \backslash A$, i.e., $\bar{\sigma} = \underline{\mathbf{4}}\backslash\{\sigma\}$, $\overline{\sigma \sigma'} = \underline{\mathbf{4}}\backslash\{\sigma,\sigma'\}$.
    We define the vertex operators $\{\mathsf{S}_{1,2,3,4}, \mathsf{Z}_K\}$ associated with $\widehat{A}_0$ quiver, obeying the following OPEs for $|q_\sigma|>1$, $\sigma \in \underline{\textbf{4}}$,
 \begin{subequations}\label{eq:OPE_A0}
    \begin{gather}
        \wick{ \c {\mathsf{S}}_\sigma(x) \c {\mathsf{S}}_\sigma(x') } = 
        \frac{(x'/x;q_\sigma^{-1})_\infty}{(q_\sigma^{-1} x'/x;q_\sigma^{-1})_\infty}
        \prod_{\sigma' \in \bar{\sigma}} 
        \frac{(q_{\overline{\sigma\sigma'}} x'/x;q_\sigma^{-1})_\infty}{(q_{\sigma'} x'/x;q_\sigma^{-1})_\infty}
        \, , \label{eq:OPE_A0_SS} \\
        \wick{ \c {\mathsf{S}}_\sigma(x) \c {\mathsf{S}}_{\sigma'}(x') } = \mathscr{S}_{\overline{\sigma\sigma'}}(q_\sigma x'/x)
        \, , \qquad 
        \wick{ \c {\mathsf{S}}_{\sigma'}(x') \c {\mathsf{S}}_\sigma(x) } = \mathscr{S}_{\overline{\sigma\sigma'}}(q_{\sigma'} x/x')
        \, , \qquad (\sigma \neq \sigma') \label{eq:OPE_A0_SS_ab} \\
        \wick{ \c {\mathsf{S}}_\sigma(x) \c {\mathsf{Z}}_K(x') } = \frac{(K x'/x;q_\sigma^{-1})_\infty}{(x'/x;q_\sigma^{-1})_\infty}
        \, , \qquad 
        \wick{ \c {\mathsf{Z}}_K(x') \c {\mathsf{S}}_\sigma(x) } = \frac{(q_\sigma^{-1} x/x';q_\sigma^{-1})_\infty}{(K^{-1} q_\sigma^{-1}x/x';q_\sigma^{-1})_\infty} \, , \\
        \mathsf{Z}_K(x) z = (K z) \mathsf{Z}_K(x) \, ,
    \end{gather}
 \end{subequations}
 with the $\mathscr{S}$ function,
 \begin{align}
     \mathscr{S}_{\sigma\sigma'}(x) = \frac{(1 - x q_{\sigma})(1 - x q_{\sigma'})}{(1 - x)(1 - x q_{\sigma\sigma'})}
     \, .
 \end{align}
 For $|q_\sigma| < 1$, we instead have
 \begin{subequations}
 \begin{gather}
        \wick{ \c {\mathsf{S}}_\sigma(x) \c {\mathsf{S}}_\sigma(x') } = \frac{(x'/x;q_\sigma)_\infty}{(q_\sigma x'/x;q_\sigma)_\infty}
        \prod_{\sigma' \in \bar{\sigma}} 
        \frac{(q_{\overline{\sigma\sigma'}}^{-1} x'/x;q_\sigma)_\infty}{(q_{\sigma'}^{-1} x'/x;q_\sigma)_\infty}
        \, , \label{eq:OPE_A0_SS2} \\
        \wick{ \c {\mathsf{S}}_\sigma(x) \c {\mathsf{Z}}_K(x') } = \frac{(q_\sigma x'/x;q_\sigma)_\infty}{(q_\sigma K x'/x;q_\sigma)_\infty}
        \, , \qquad 
        \wick{ \c {\mathsf{Z}}_K(x') \c {\mathsf{S}}_\sigma(x) } = \frac{(K^{-1} x/x';q_\sigma)_\infty}{(x/x';q_\sigma)_\infty} \, .
 \end{gather}
 \end{subequations}
\end{definition}
The free field realization of these vertex operators is given in \eqref{eq:free_field}.
Compared with the $\mathsf{Z}$-operator shown in \eqref{eq:D8_VO}, we have $\mathsf{Z}_K(x) = \mathsf{Z}(Kx)/\mathsf{Z}(x)$.
In this case, we do not impose the non-trivial relation between $z$ and $\mathsf{S}_\sigma$ as in \eqref{eq:OPE_z}.
We remark that $\mathsf{S}_\sigma$ and $\mathsf{S}_{\sigma'}$ commute if $\sigma \neq \sigma'$.
See \eqref{eq:OPE_A0_SS_ab}.
    
\begin{remark}
    The vertex operator $\mathsf{Z}$ would be identified with (the rank-one case of) $\phi_i$ in \cite[\S2.3]{Shiraishi:2019JIS}, associated with the framing spaces of the affine Laumon space.
    Moreover, (the rank-one case of) their screening current $S_i$ agrees with (one of) ours, $\mathsf{S}$.
\end{remark}

\begin{lemma}\label{lemma:SS_A0^}
 For $\sigma \in \underline{\mathbf{4}}$, we write $\bar{\sigma} = \underline{\mathbf{4}} \backslash \sigma$.
 For $|q_\sigma| > 1$, we may rewrite the screening current OPE as follows,
 \begin{subequations}
  \begin{align}
   \wick{ \c {\mathsf{S}}_\sigma(x) \c {\mathsf{S}}_\sigma(x') } & = \frac{((x'/x)^{\pm 1};q_\sigma^{-1})_\infty (q_{\sigma'\sigma''} (x'/x)^{\pm 1};q_\sigma^{-1})_\infty}{(q_{\sigma'} (x'/x)^{\pm 1};q_\sigma^{-1})_\infty (q_{\sigma''} (x'/x)^{\pm 1};q_\sigma^{-1})_\infty} \frac{\theta(q_{\sigma'} x'/x;q_\sigma^{-1}) \theta(q_{\sigma''} x'/x;q_\sigma^{-1})}{\theta(x'/x;q_\sigma^{-1}) \theta(q_{\sigma'\sigma''} x'/x;q_\sigma^{-1})}
   \, ,
   \qquad (\sigma',\sigma'' \in \bar{\sigma}) \\
   & = \frac{((x'/x)^{\pm 1};q_\sigma^{-1})_\infty }{\prod_{\sigma' \in \bar{\sigma}}(q_{\sigma'} (x'/x)^{\pm 1};q_\sigma^{-1})_\infty } \frac{\prod_{\sigma' \in \bar{\sigma}} \theta(q_{\sigma'} x'/x;q_\sigma^{-1})}{\theta(x'/x;q_\sigma^{-1})} \, .
  \end{align}
 \end{subequations}
\end{lemma}

\subsection{Single screening current}

In this subsection, we set $\mathsf{S}(x) \equiv \mathsf{S}_4(x)$ and assume $|q_4| > 1$.
We apply the same definition for the screened vertex operator as before (Definition~\ref{def:Z_screened}).

\begin{proposition}\label{prop:conf_block_A0^}
    We have the two different integral formulas for the conformal block in the chamber $\mathscr{C}_{12\ldots n}: a_1 \succ a_2 \succ \cdots \succ a_n$,
    \begin{subequations}
    \begin{align}
        \Psi_{\mathscr{C}_{1\ldots n}}^{\underline{k}}(\underline{a};\underline{K}) & = \bra{0} \mathsf{Z}^{(k_1)}_{K_1}(a_1) \mathsf{Z}^{(k_2)}_{K_2}(a_2) \cdots \mathsf{Z}^{(k_n)}_{K_n}(a_n) \ket{0} / \bra{0} \mathsf{Z}_{K_1}(a_1) \mathsf{Z}_{K_2}(a_2) \cdots \mathsf{Z}_{K_n}(a_n) \ket{0}
        \nonumber \\ 
        & = \frac{\mathsf{c}_{4,\underline{k}}(\underline{K})}{k!} \oint_\gamma [\dd{\underline{x}}] \, \Phi_{4;\sigma'}^\text{F}(\underline{x},\underline{a})\widetilde{\operatorname{Stab}}_{4;\sigma',\underline{k}}^\text{F}
        (\underline{x},\underline{a}) \, , \label{eq:block_A0^_F} \\
        & = \frac{\mathsf{c}_{4,\underline{k}}(\underline{K})}{k!} \oint_\gamma [\dd{\underline{x}}] \, \Phi_{4}^\text{M}(\underline{x},\underline{a})\widetilde{\operatorname{Stab}}_{4,\underline{k}}^\text{M}
        (\underline{x},\underline{a}) \label{eq:block_A0^_M} \, ,
    \end{align}
    \end{subequations}
    where the principal parts are defined by
    \begin{subequations}\label{eq:principal_part_A0^}
        \begin{align}
            \Phi_{\sigma;\sigma'}^\text{F}(\underline{x},\underline{a}) & = \prod_{i \neq j}^k \frac{(x_{ij};q_\sigma^{-1})_\infty (q_{\overline{\sigma \sigma'}} x_{ij};q_\sigma^{-1})_\infty}{\prod_{\sigma'' \in \overline{\sigma \sigma'}}(q_{\sigma''} x_{ij};q_\sigma^{-1})_\infty} \prod_{\substack{i=1,\ldots,k \\ \alpha=1,\ldots,n}} \frac{(K_\alpha a_\alpha/x_i;q_\sigma^{-1})_\infty}{( a_\alpha/x_i;q_\sigma^{-1})_\infty}
            \, , \\
            \Phi_\sigma^\text{M}(\underline{x},\underline{a}) & = \prod_{i \neq j}^k \frac{(x_{ij};q_\sigma^{-1})_\infty}{\prod_{\sigma' \in \overline{\sigma }}(q_{\sigma'} x_{ij};q_\sigma^{-1})_\infty} \prod_{\substack{i=1,\ldots,k \\ \alpha=1,\ldots,n}} \frac{(K_\alpha a_\alpha/x_i;q_\sigma^{-1})_\infty}{( a_\alpha/x_i;q_\sigma^{-1})_\infty} \, ,
        \end{align}
    \end{subequations}
    with the constant term under a slight modification from \eqref{eq:c_const},
    \begin{align}
        \mathsf{c}_{\sigma,\underline{k}}(\underline{K}) = \prod_{i=1}^k \frac{\theta(K_{c(i)};q_\sigma^{-1})}{\theta(z K_i^c;q_\sigma^{-1})} \, .
        \label{eq:c_const2}
    \end{align}
    We define two different off-shell elliptic (pre-)stable envelopes as follows,
\begin{subequations}
    \begin{align}
        \widetilde{\operatorname{Stab}}_{\sigma;\sigma',\underline{k}}^\text{F}(\underline{x},\underline{a}) 
        & = \operatorname{Sym} \prod_{i<j}^k \frac{\prod_{\sigma''\in\overline{\sigma\sigma'}} \theta(q_{\sigma''} x_{ji};q_\sigma^{-1})}{\theta(x_{ji};q_\sigma^{-1})\theta(q_{\overline{\sigma\sigma'}} x_{ji};q_\sigma^{-1})} \nonumber \\
        & \qquad \times \prod_{i=1}^k \frac{\prod_{\alpha < c(i)} \theta(a_\alpha / x_i ; q_\sigma^{-1}) \theta(z K_i^c a_{c(i)} / x_i ; q_\sigma^{-1}) \prod_{\alpha > c(i)} \theta(K_\alpha a_\alpha / x_i ; q_\sigma^{-1})}{\prod_{\alpha=1}^k \theta(K_\alpha a_\alpha / x_i ; q_\sigma^{-1})} 
        \, , \\
        \widetilde{\operatorname{Stab}}_{\sigma,\underline{k}}^\text{M}(\underline{x},\underline{a}) 
        & = \operatorname{Sym} \prod_{i<j}^k \frac{\theta(q_{1,2,3} x_{ji};q_\sigma^{-1})}{\theta(x_{ji};q_\sigma^{-1})} \nonumber \\
        & \qquad \times \prod_{i=1}^k \frac{\prod_{\alpha < c(i)} \theta(a_\alpha / x_i ; q_\sigma^{-1}) \theta(z K_i^c a_{c(i)} / x_i ; q_\sigma^{-1}) \prod_{\alpha > c(i)} \theta(K_\alpha a_\alpha / x_i ; q_\sigma^{-1})}{\prod_{\alpha=1}^k \theta(K_\alpha a_\alpha / x_i ; q_\sigma^{-1})} \, .
    \end{align}
\end{subequations}
\end{proposition}
\begin{proof}
    We may apply the same analysis as discussed in Proposition~\ref{prop:conf_block_A1}.
    In this case, we instead use the OPEs given in Definition~\ref{def:OPEs_A0} together with Lemma~\ref{lemma:SS_A0^} to obtain this formula.
\end{proof}
For the moment, the choice of $\sigma' \in \bar{4}$ in $\Phi_{4;\sigma'}^\text{F}$ and $\widetilde{\operatorname{Stab}}^\text{F}_{4;\sigma'}$ is not unique.
It would be determined by the specialization of the parameter $K_\alpha = q_{\sigma'}$.
The conformal block in the ``Fock formula''~\eqref{eq:block_A0^_F} is identified as the Coulomb branch localization formula of the 3d ADHM sigma model~\cite{Choi:2019zpz,Crew:2020psc}, where the corresponding matter content is described by the quiver diagram given in Fig.~\ref{fig:single_F}.
Namely, it turns out to be the vertex function associated with the quasimaps to the Hilbert scheme of points on $\mathbb{C}^2$ denoted by $\operatorname{Hilb}_k(\mathbb{C}^{2})$ for $n = 1$~\cite{Smirnov:2019rmq}, and the corresponding Quot scheme for $n > 1$.
In the ``MacMahon formula''~\eqref{eq:block_A0^_M}, it turns out to be similarly the vertex function for the quasimaps to $\operatorname{Hilb}_k(\mathbb{C}^{3})$~\cite{Cao:2023lon,Piazzalunga:2023qik} and also the higher-rank analog.
See the corresponding quiver in Fig.~\ref{fig:single_M}.
One can realize two different forms based on the same vertex operator as observed before~\cite{Kimura:2023bxy}.

The stable envelopes $\widetilde{\operatorname{Stab}}^\text{F,M}$ are associated with the degree $n$ tensor product of Fock and the MacMahon modules of quantum toroidal $\mathfrak{gl}_1$ as discussed below.
We remark that $\widetilde{\operatorname{Stab}}^\text{F,M}$ defined above are not yet precisely the stable envelopes:
The expression of $\widetilde{\operatorname{Stab}}^\text{F}$ does not agree with that given by Smirnov~\cite{Smirnov:2018drm} (in the case $n = 1$, hence associated with $\operatorname{Hilb}_k(\mathbb{C}^{2})$) although they have a close form.
The stable envelope for $\mathbb{C}^{3}$ is not yet even constructed.
For general $(n,k)$, the stable envelopes for $\mathbb{C}^2$ and $\mathbb{C}^3$ should be parametrized by $n$-tuple 2d and 3d partitions of size $k$.
For this purpose, it would be plausible to update the definition of the screened vertex operator (Definition~\ref{def:Z_screened}) in order that it depends on 2d/3d partitions.
We may need the tree structure as discussed in \cite{Smirnov:2018drm} for the case $\mathbb{C}^2$.
In order to emphasize this point, we would also call $\widetilde{\operatorname{Stab}}^\text{F,M}$ the ``pre-stable envelopes''.

\subsubsection{Fock formula}

\paragraph{Pole structure}
Let $x_{\alpha,i} = x_{i+\sum_{\beta > \alpha}k_\beta}$ for $i = 1,\ldots, k_\alpha$ as before.
The poles of the contour integral are parametrized by $n$-tuple partitions $\underline{\lambda} = (\lambda_1,\ldots,\lambda_n)$, hence corresponding to the $n$-tensor product of the Fock modules of quantum toroidal $\mathfrak{gl}_1$,
\begin{align}
 x_{\alpha,1} = a_\alpha q_4^{-n_{\alpha,1}} 
 \, , \quad 
 x_{\alpha,i} = 
 x_{\alpha,j(<i)} q_{1} q_4^{-n_{\alpha,i}} 
 \quad \text{or} \quad 
 x_{\alpha,j(<i)} q_{2} q_4^{-n_{\alpha,i}} 
 \, .
\end{align} 
Hence, we have
\begin{align}
 \underline{x}(\underline{\pi}) = \{ x_{\alpha,i} \mid x_{\alpha,i} = a_\alpha q_1^{i_1-1} q_2^{i_2-1} q_4^{-\pi_{\alpha,i}}, i = (i_1,i_2) \in \lambda_\alpha \} \, ,
    \label{eq:poles_PT1}
\end{align} 
where we write the position of the $i$-th coordinate in the partition $\lambda_\alpha$ by $i = (i_1, i_2)$ for simplicity.
As pointed out in \cite{Crew:2020psc}, the configuration denoted by $\underline{\pi} = \{ \pi_{\alpha,i}, i = (i_1,i_2) \in \lambda_\alpha \}_{\alpha = 1,\ldots,n}$ is an $n$-tuple reverse plane partition whose base is fixed by the partitions $\underline{\lambda} = \{\lambda_\alpha\}$.
We denote by $\mathsf{PP}_{\underline{\lambda}}$ the set of such $n$-tuple reverse plane partitions.
Similarly to Propositions~\ref{prop:Higgs1} and \ref{prop:Higgs2}, we should impose $K_\alpha = q_3$ to realize this pole structure.
Namely we take $\Phi_{4;3}^\text{F}$ for the principal part of the conformal block.
It has been known that this box counting structure agrees with the one-leg PT3 count.

\paragraph{Vertex function}

In order to deal with the conformal block, we introduce the combinatorial notations.
Let $\lambda_\alpha$ be a partition and $i = (i_1,i_2) \in \lambda_\alpha$ for $\alpha = 1, \ldots, n$.
We denote by $\lambda_\alpha^\rmT$ the transposition of $\lambda_\alpha$.
We define the arm length and the leg length as follows,
    \begin{align}
        \textsl{a}_\alpha(i) = \lambda_{\alpha,i_1} - i_2
        \, , \qquad 
        \textsl{l}_\alpha(i) = \lambda_{\alpha,i_2}^\rmT - i_1
        \, .
    \end{align}
We have the following formula.
\begin{lemma}
    Let $\emptyset_{\underline{\lambda}} \in \mathsf{PP}_{\underline{\lambda}}$ be the empty configuration, and we write $(\cdot)_\infty = (\cdot;q_4^{-1})_\infty$.
    Then, we have the following combinatorial formula for the contour integral,
    \begin{align}
    I_{\underline{\lambda}} (\underline{a}) 
    & = \frac{(q_4^{-1})_\infty^k}{k!} \oint_{\emptyset_{\underline{\lambda}}} [\dd{\underline{x}}] \prod_{i \neq j}^k (x_{ij})_\infty \prod_{i,j}^k \frac{(q_{12} x_{ij})_\infty}{(q_{1,2} x_{ij})_\infty} \prod_{\substack{\alpha = 1,\ldots,n \\ i = 1, \ldots, k}} \frac{1}{(a_\alpha / x_i)_\infty(q_{12} x_i / a_\alpha)_\infty}
    \nonumber \\ &  
    = \prod_{1 \le \alpha,\beta \le n} \left[ \prod_{i \in \lambda_\beta} (a_{\beta \alpha} q_1^{\textsl{l}_\alpha(i) + 1} q_2^{-\textsl{a}_\beta(i)})^{-1}_\infty \prod_{i \in \lambda_\alpha} (a_{\beta \alpha} q_1^{-\textsl{l}_\beta(i)} q_2^{\textsl{a}_\alpha(i)+1})^{-1}_\infty \right]
    \, .
\end{align}
\end{lemma}
\begin{proof}
    This can be shown in the same way as Lemma~\ref{lemma:A1_pert}.
    We rewrite the contour integral as the equivariant index,    
    \begin{align}
        I_{\underline{\lambda}} (\underline{a}) = \operatorname{ch} \wedge \left[ \frac{1}{\mathbf{P}_4^\vee} \left( \mathbf{P}_{12} \mathbf{K} \mathbf{K}^\vee - \mathbf{N} \mathbf{K}^\vee - \mathbf{Q}_{12} \mathbf{K} \mathbf{N}^\vee \right) \right]_{\emptyset_{\underline{\lambda}}} \, ,
    \end{align}
    where $\operatorname{ch} \mathbf{P}_{A} = \prod_{a \in A} (1-q_a)$ and $\operatorname{ch} \mathbf{Q}_{A} = \prod_{a \in A} q_{a}$ for $A \subseteq \underline{\mathbf{4}}$.
    Noticing 
    \begin{align}
        \operatorname{ch} \left[ \mathbf{P}_{12} \mathbf{K} \mathbf{K}^\vee - \mathbf{N} \mathbf{K}^\vee - \mathbf{Q}_{12} \mathbf{K} \mathbf{N}^\vee \right]_{\emptyset_{\underline{\lambda}}} = - \sum_{1 \le \alpha, \beta \le n} a_{\beta \alpha} \left( \sum_{i \in \lambda_\beta} q_1^{\textsl{l}_\alpha(i)+1} q_2^{-\textsl{a}_\beta(i)} + \sum_{i \in \lambda_\alpha} q_1^{-\textsl{l}_\beta(i)} q_2^{\textsl{a}_\alpha(i)+1} \right) \, ,
    \end{align}
    (see, e.g., \cite{Nakajima:1999,Kimura:2020jxl} for details) and then evaluating the index, we obtain the formula.
\end{proof}
This is the $q$-shifted version of the Nekrasov partition function of $\mathrm{U}(n)$ gauge theory on $\mathbb{C}_{q_1} \times \mathbb{C}_{q_2} \times \mathbb{S}^1$~\cite{Nekrasov:2002qd,Nekrasov:2003rj}, which is given by the equivariant integral over the instanton moduli space isomorphic to the Hilbert scheme of points on $\mathbb{C}^2$ and its higher-rank analog.

We evaluate the contour integral with the poles given in \eqref{eq:poles_PT1}.
The principal part of the conformal block is given as follows.
\begin{theorem}\label{thm:vertex_Fock}
    For $i = (i_1,i_2)$ and $j = (j_1,j_2)$, we write $q_{12}^i = q_{12}^{(i_1,i_2)} = q_1^{i_1} q_2^{i_2}$.
    We write $(\cdot)_m = (\cdot;q_4^{-1})_m$ and $\theta(\cdot) = \theta(\cdot;q_4^{-1})$.
    We put $K_\alpha = q_3$ for $\alpha = 1, \ldots, n$.
    Then, the contribution of the principal part of the conformal block associated with the poles specified in \eqref{eq:poles_PT1} is given as follows,
    \begin{align}
    \frac{1}{k!} \oint_{\{\underline{\pi}\}} [\dd{\underline{x}}] \, \Phi_{4;3}^\text{F}(\underline{x},\underline{a}) 
    = \mathring{Z}_{\underline{\lambda}} \sum_{\underline{\pi} \in \mathsf{PP}_{\underline{\lambda}}} Z_{\underline{\pi}}
    \end{align}
    where
\begin{subequations}
\begin{align}
    \mathring{Z}_{\underline{\lambda}} 
    & = \frac{(q_{1,2})_\infty^k}{(q_4^{-1})_\infty^k (q_{12})_\infty^k} \prod_{\substack{1 \le \alpha, \beta \le n \\ i \in \lambda_\alpha}} \theta(q_3 a_{\beta \alpha} q_i^{-1}) \prod_{1 \le \alpha,\beta \le n} \left[ \prod_{i \in \lambda_\beta} (a_{\beta \alpha} q_1^{\textsl{l}_\alpha(i) + 1} q_2^{-\textsl{a}_\beta(i)})^{-1}_\infty \prod_{i \in \lambda_\alpha} (a_{\beta \alpha} q_1^{-\textsl{l}_\beta(i)} q_2^{\textsl{a}_\alpha(i)+1})^{-1}_\infty \right]
    \, , \\
    Z_{\underline{\pi}} & = \prod_{(\alpha,i) \neq (\beta,j)}
    \frac{(q_{1} a_{\beta\alpha} q_{12}^j / q_{12}^i,q_{2} a_{\beta\alpha} q_{12}^j / q_{12}^i)_{\pi_{\beta,j}-\pi_{\alpha,i}}}{(a_{\beta\alpha} q_{12}^j / q_{12}^i,q_{12} a_{\beta\alpha} q_{12}^j / q_{12}^i)_{\pi_{\beta,j}-\pi_{\alpha,i}}} \prod_{\substack{\alpha,\beta=1,\ldots,n \\ i \in \lambda_\alpha}} \frac{(q_{12} a_{\beta\alpha} / q_{12}^{i})_{-\pi_{\alpha,i}}}{(q_{123} a_{\beta\alpha} / q_{12}^{i})_{-\pi_{\alpha,i}}}
    \, .
\end{align}    
\end{subequations}
\end{theorem}
This is a higher-rank version of the vertex function associated with the Hilbert scheme of points on $\mathbb{C}^2$ given by \cite{Smirnov:2019rmq,Crew:2020psc}.
See also \cite{Choi:2019zpz}.
In this expression, there is no counting parameter in the summation over $\underline{\pi} \in \mathsf{PP}_{\underline{\lambda}}$, which can be added by hand.
In fact, the stable envelope gives rise to the factor $z^{|\underline{\pi}|}$ in the summation, which is a consequence of the quasiperiodicity of the kernel function $\varphi_K(\cdot,\cdot)$ as seen in Proposition~\ref{prop:vort_fn}.

\subsubsection{MacMahon formula}

\paragraph{Pole structure}
Let $x_{\alpha,i} = x_{i+\sum_{\beta > \alpha}k_\beta}$ for $i = 1,\ldots, k_\alpha$ as before.
The poles of the contour integral are parametrized by $n$-tuple plane partitions $\underline{\pi} = (\pi_1,\ldots,\pi_n)$,
\begin{align}
 x_{\alpha,1} = a_\alpha q_4^{-n_{\alpha,1}} 
 \, , \quad 
 x_{\alpha,i} = 
 x_{\alpha,j(<i)} q_{1} q_4^{-n_{\alpha,i}} 
 \quad \text{or} \quad 
 x_{\alpha,j(<i)} q_{2} q_4^{-n_{\alpha,i}} 
  \quad \text{or} \quad 
 x_{\alpha,j(<i)} q_{3} q_4^{-n_{\alpha,i}} 
 \, .
\end{align} 
Hence, we have
\begin{align}
 \underline{x}(\underline{\rho}) = \{ x_{\alpha,i} \mid x_{\alpha,i} = a_\alpha q_1^{i_1-1} q_2^{i_2-1} q_3^{i_3-1} q_4^{-\rho_{\alpha,i}}, i = (i_1,i_2,i_3) \in \pi_\alpha \} \, , \quad \underline{\rho} \in \mathsf{SP}_{\underline{\pi}} \, ,
 \label{eq:poles_PT4one}
\end{align} 
where $\mathsf{SP}_{\underline{\pi}} := \{\rho_{\alpha,ijk}, i = (i_1,i_2,i_3) \in \pi_\alpha \}_{\alpha = 1,\ldots,n}$ is the set of $n$-tuple reverse solid partition with the fixed base given by the plane partition $\underline{\pi} = \{\pi_\alpha\}$.

\begin{proposition}
    Let $K_\alpha = q_1^{\mu^{(1)}_{\alpha}} q_2^{\mu^{(2)}_\alpha} q_3^{\mu^{(3)}_\alpha}$, $\mu^{(1,2,3)}_\alpha \in \mathbb{Z}_{\ge 0}$ for $\alpha = 1,\ldots,n$.
    Then, $Z_{\underline{\pi}} = 0$ unless $(\mu^{(1)}_\alpha + 1,\mu^{(2)}_\alpha + 1,\mu^{(3)}_\alpha + 1) \not\in \pi_\alpha$.
\end{proposition}

\paragraph{Vertex function}

The contour integral evaluated with the poles \eqref{eq:poles_PT4one} gives rise to the following vertex function.
We write $q_{123}^{i} = q_{123}^{(i_1,i_2,i_3)} = q_1^{i_1} q_2^{i_2} q_3^{i_3}$ for $i = (i_1,i_2,i_3)$.
\begin{theorem}\label{thm:vertex_MacMahon}
    We write $(\cdot)_m = (\cdot;q_4^{-1})_m$.
    The contribution of the principal part of the conformal block associated with the poles specified in \eqref{eq:poles_PT4one} is given as follows,
    \begin{align}
    \frac{1}{k!} \oint_{\{\underline{\rho}\}} [\dd{\underline{x}}] \, \Phi_{4}^\text{M}(\underline{x},\underline{a}) 
    = \mathring{Z}_{\underline{\pi}} \sum_{\underline{\rho} \in \mathsf{SP}_{\underline{\pi}}} Z_{\underline{\rho}}
    \end{align}
    where
\begin{subequations}
\begin{align}
    \mathring{Z}_{\underline{\pi}} & = \frac{(q_{1,2,3})_\infty^k}{(q_4^{-1})_\infty^k} \operatorname{ch} \wedge \left[ \frac{1}{\mathbf{P}_{4}^\vee} \left( (1 - \mathbf{Q}_{1,2,3} ) \mathbf{K} \mathbf{K}^\vee - \mathbf{N} \mathbf{K}^\vee + \widetilde{\mathbf{N}} \mathbf{K}^\vee \right) \right]_{\emptyset_{\underline{\pi}}}
    \, , \\
    Z_{\underline{\rho}} & = \prod_{(\alpha,i) \neq (\beta,j)}
    \frac{(q_{1,2,3} a_{\beta\alpha}q_{123}^j/q_{123}^i)_{\rho_{\beta,j}-\rho_{\alpha,i}}}{(a_{\beta\alpha}q_{123}^j/q_{123}^i)_{\rho_{\beta,j}-\rho_{\alpha,i}}} \prod_{\substack{\alpha,\beta=1,\ldots,n \\ i \in \pi_\alpha}} \frac{(a_{\beta\alpha} / q_{123}^{i})_{-\rho_{\alpha,i}}}{(K_\beta a_{\beta\alpha} / q_{123}^{i})_{-\rho_{\alpha,i}}}
    \, .
\end{align}    
\end{subequations}
\end{theorem}
This is a higher-rank version of the vertex function associated with the Hilbert scheme of points on $\mathbb{C}^3$ given by  \cite{Cao:2023lon,Piazzalunga:2023qik}, i.e., the Quot scheme of $\mathbb{C}^3$, as this formula is also applicable to higher-rank framing bundle $\mathbf{N}$ and $\widetilde{\mathbf{N}}$.
We remark that the perturbative factor $\mathring{Z}_{\underline{\pi}}$ is consistent with the edge contribution shown in~\cite{Nekrasov:2023nai} up to the boundary factor.
Again, there is no counting factor $z^{|\underline{\rho}|}$ in the summation over $\underline{\rho} \in \mathsf{SP}_{\underline{\pi}}$, which can be imposed by the stable envelope thanks to the quasiperiodicity of the kernel function $\varphi_K(\cdot,\cdot)$.

\subsection{Multiple screening current}

We consider the conformal block of the multi-screened vertex operators, and discuss its relation to the enumerative invariants.

\begin{definition}
    Let $\underline{m} = \{m_\sigma\}_{\sigma\in\underline{\mathbf{4}}} \in \mathbb{Z}_{\ge 0}^4$.
    We define the $\underline{m}$-screened vertex operator,
    \begin{align}
        & \mathsf{Z}_K^{(\underline{m})}(a) 
        \nonumber \\ & 
        = \oint_\gamma [\dd{\underline{x}}] \prod_{\sigma\in\underline{\mathbf{4}}, |q_\sigma| > 1}
        \left( \prod_{i=1,\ldots,m_\sigma}^{\curvearrowleft} \varphi_{\sigma,K}\left(\frac{a}{x_{\sigma,i}},zK \right) \mathsf{S}_\sigma(x_{\sigma,i}) \right)  \mathsf{Z}_K(a) \prod_{\sigma\in\underline{\mathbf{4}}, |q_\sigma| < 1}
        \left( \prod_{i=1,\ldots,m_\sigma}^{\curvearrowright} \varphi_{\sigma,K}\left(\frac{a}{x_{\sigma,i}},zK \right) \mathsf{S}_\sigma(x_{\sigma,i}) \right)  
        \, .
    \end{align}
\end{definition}
When we consider the quadruple screened case ($m_\sigma > 0$ for all $\sigma \in \underline{\mathbf{4}}$), one cannot take all the parameters $|q_\sigma| > 1$ nor $|q_\sigma| < 1$ as a consequence of the condition~\eqref{eq:CY4_cond}.
Hence, we should consider both situations in this case.
For the other cases, we can take all the parameters $|q_\sigma| > 1$ or $|q_\sigma| < 1$ for simplicity.

\begin{figure}[t]
    \centering
    \begin{subfigure}[b]{0.17\textwidth}
         \centering
         \begin{tikzpicture}
             \node[circle,draw=black,minimum size=5pt,inner sep=1.5pt] (k) at (0,0) {$k$};
             \node[rectangle,draw=black] (n1) at (.5,1) {$n$};
             \node[rectangle,draw=black] (n2) at (-.5,1) {$n$};
             \draw [->-] (n1) -- (k);
             \draw [-<-] (n2) -- (k);
             \draw (k) to [bend right = 45,in=260] ++(0,-1) to [bend right = 45,out=-80] (k);
             \draw [-latex] (0,-1) -- ++ (.1,0);
             \draw (k) to [bend right = 45,in=260] ++(0,-1.2) to [bend right = 45,out=-80] (k);
             \draw [-latex] (0,-1.2) -- ++ (.1,0);
             \draw [densely dashed] (k) to [bend right = 45,in=260] ++(0,-.8) to [bend right = 45,out=-80] (k);
             \draw [-latex] (0,-.8) -- ++ (.1,0);
         \end{tikzpicture}
         \caption{Single case (F)}
         \label{fig:single_F}
     \end{subfigure}
     \hfill
         \begin{subfigure}[b]{0.17\textwidth}
         \centering
         \begin{tikzpicture}
             \node[circle,draw=black,minimum size=5pt,inner sep=1.5pt] (k) at (0,0) {$k$};
             \node[rectangle,draw=black] (n1) at (.5,1) {$n$};
             \node[rectangle,draw=black] (n2) at (-.5,1) {$n$};
             \draw [->-] (n1) -- (k);
             \draw [-<-] (n2) -- (k);
             \draw (k) to [bend right = 45,in=260] ++(0,-1) to [bend right = 45,out=-80] (k);
             \draw [-latex] (0,-1) -- ++ (.1,0);
             \draw (k) to [bend right = 45,in=260] ++(0,-1.2) to [bend right = 45,out=-80] (k);
             \draw [-latex] (0,-1.2) -- ++ (.1,0);
             \draw (k) to [bend right = 45,in=260] ++(0,-.8) to [bend right = 45,out=-80] (k);
             \draw [-latex] (0,-.8) -- ++ (.1,0);
         \end{tikzpicture}
         \caption{Single case (M)}
         \label{fig:single_M}
     \end{subfigure}
     \hfill
    \begin{subfigure}[b]{0.32\textwidth}
         \centering
         \begin{tikzpicture}[scale=.9]
             \node[circle,draw=black,minimum size=5pt,inner sep=1.5pt] (k1) at (-1.5,0) {$k_1$};
             \node[circle,draw=black,minimum size=5pt,inner sep=1.5pt] (k2) at (0,-1) {$k_2$};
             \node[circle,draw=black,minimum size=5pt,inner sep=1.5pt] (k3) at (1.5,0) {$k_3$};    
             \node[rectangle,draw=black] (n1) at (.5,1) {$n$};
             \node[rectangle,draw=black] (n2) at (-.5,1) {$n$};
             \foreach \x in {1,2,3}{
             \draw [->-] (n1) -- (k\x);
             \draw [-<-] (n2) -- (k\x);
             }
             \draw (k1) to [bend right = 45,in=260] ++(-1,0) to [bend right = 45,out=-80] (k1);
             \draw [-latex] (k1)++(-1,0) -- ++ (0,.1);
             \draw (k1) to [bend right = 45,in=260] ++(-1.2,0) to [bend right = 45,out=-80] (k1);
             \draw [-latex] (k1)++(-1.2,0) -- ++ (0,.1);
             \draw [densely dashed] (k1) to [bend right = 45,in=260] ++(-.8,0) to [bend right = 45,out=-80] (k1);
             \draw [-latex] (k1)++(-.8,0) -- ++ (0,.1);     
             \draw (k2) to [bend right = 45,in=260] ++(0,-1) to [bend right = 45,out=-80] (k2);
             \draw [-latex] (k2)++(0,-1) -- ++ (.1,0);
             \draw (k2) to [bend right = 45,in=260] ++(0,-1.2) to [bend right = 45,out=-80] (k2);
             \draw [-latex] (k2)++(0,-1.2) -- ++ (.1,0);
             \draw [densely dashed] (k2) to [bend right = 45,in=260] ++(0,-.8) to [bend right = 45,out=-80] (k2);
             \draw [-latex] (k2)++(0,-.8) -- ++ (.1,0);
             \draw (k3) to [bend right = 45,in=260] ++(1,0) to [bend right = 45,out=-80] (k3);
             \draw [-latex] (k3)++(1,0) -- ++ (0,.1);
             \draw (k3) to [bend right = 45,in=260] ++(1.2,0) to [bend right = 45,out=-80] (k3);
             \draw [-latex] (k3)++(1.2,0) -- ++ (0,.1);
             \draw [densely dashed] (k3) to [bend right = 45,in=260] ++(.8,0) to [bend right = 45,out=-80] (k3);
             \draw [-latex] (k3)++(.8,0) -- ++ (0,.1);          
             \draw [magenta,double] (k1) -- (k2);
             \draw [magenta,double] (k1) -- (k3);
             \draw [magenta,double] (k2) -- (k3);
         \end{tikzpicture}
         \caption{Triple case (F)}
         \label{fig:triple_F}
     \end{subfigure}     
     \hfill
    \begin{subfigure}[b]{0.32\textwidth}
         \centering
         \begin{tikzpicture}[scale=.9]
             \node[circle,draw=black,minimum size=5pt,inner sep=1.5pt] (k1) at (-1.5,0) {$k_1$};
             \node[circle,draw=black,minimum size=5pt,inner sep=1.5pt] (k2) at (-.7,-1) {$k_2$};
             \node[circle,draw=black,minimum size=5pt,inner sep=1.5pt] (k3) at (.7,-1) {$k_3$};             
             \node[circle,draw=black,minimum size=5pt,inner sep=1.5pt] (k4) at (1.5,0) {$k_4$};    
             \node[rectangle,draw=black] (n1) at (.5,1) {$n$};
             \node[rectangle,draw=black] (n2) at (-.5,1) {$n$};
             \foreach \x in {1,2,3,4}{
             \draw [->-] (n1) -- (k\x);
             \draw [-<-] (n2) -- (k\x);
             }
             \foreach \x in {.8,1,1.2} {
             \draw (k1) to [bend right = 45,in=260] ++(-\x,0) to [bend right = 45,out=-80] (k1);
             \draw [-latex] (k1)++(-\x,0) -- ++ (0,.1);
             }
             \foreach \y in {2,3} {
             \foreach \x in {.8,1,1.2} {
             \draw (k\y) to [bend right = 45,in=260] ++(0,-\x) to [bend right = 45,out=-80] (k\y);
             \draw [-latex] (k\y)++(0,-\x) -- ++ (.1,0);
             } }
             \foreach \x in {.8,1,1.2} {
             \draw (k4) to [bend right = 45,in=260] ++(\x,0) to [bend right = 45,out=-80] (k4);
             \draw [-latex] (k4)++(\x,0) -- ++ (0,.1);
             }             
             \draw [magenta,double] (k1) -- (k2);
             \draw [magenta,double] (k1) -- (k3);
             \draw [magenta,double] (k1) -- (k4);
             \draw [magenta,double] (k2) -- (k3);
             \draw [magenta,double] (k2) -- (k4);
             \draw [magenta,double] (k3) -- (k4);
         \end{tikzpicture}
         \caption{Quadruple case (M)}
         \label{fig:quadruple_M}
     \end{subfigure}          
    \caption{Quivers for $\mathrm{W}_{q_{1,2,3,4}}(\widehat{A}_0)$ block. 
    (F) and (M) stand for the Fock and the MacMahon formula, respectively.
    The solid (dashed) lines indicate the chiral (Fermi) multiplets, and the red ones are multiple 0d theory contributions.}
    \label{fig:quiver_A0^}
\end{figure}
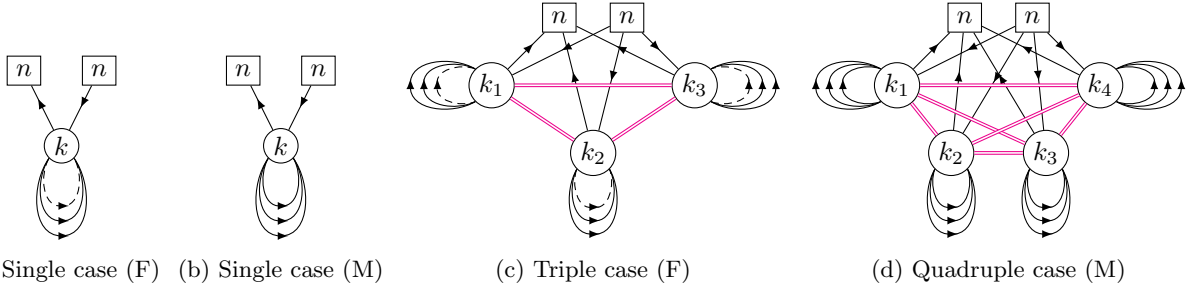

\paragraph{Fock formula}

We take $K_\alpha = q_4$ and $\underline{k}_\alpha = (k_{\alpha,1},k_{\alpha,2},k_{\alpha,3},0)$ for all $\alpha = 1,\ldots,n$ to consider the Fock formula.
We assume $|q_{1,2,3}|>1$.
Then, we have the contour integral formula of the conformal block as follows,
\begin{align}
    \Psi_{\mathscr{C}_{1\ldots n}}^{\underline{k}}(\underline{a};q_4) & = \bra{0} \mathsf{Z}^{(\underline{k}_1)}_{q_4}(a_1) \mathsf{Z}^{(\underline{k}_2)}_{q_4}(a_2) \cdots \mathsf{Z}^{(\underline{k}_n)}_{q_4}(a_n) \ket{0} / \bra{0} \mathsf{Z}_{q_4}(a_1) \mathsf{Z}_{q_4}(a_2) \cdots \mathsf{Z}_{q_4}(a_n) \ket{0}
    \nonumber \\ 
    & = \prod_{\sigma \in \underline{\mathbf{3}}} \frac{\mathsf{c}_{\sigma,\underline{k}}}{k_\sigma!} \oint_\gamma [\dd{\underline{x}}] \prod_{\sigma \in \underline{\mathbf{3}}} \Phi_{\sigma;4}^\text{F}(\underline{x}_\sigma,\underline{a}) \widetilde{\operatorname{Stab}}_{\sigma;4,\underline{k}}^\text{F}(\underline{x}_\sigma,\underline{a}) 
    \prod_{\sigma, \sigma' \in \underline{\mathbf{3}}} \prod_{\substack{i = 1,\ldots,k_\sigma \\ j = 1, \ldots, k_{\sigma'}}} \mathscr{S}_{\overline{\sigma\sigma'}}(q_\sigma x_{\sigma',j}/x_{\sigma,i})
    \, . \label{eq:block_F_multi}
\end{align}
The poles of this contour integral are parametrized by $n$-tuple three-dimensional partitions with a triple of boundary two-dimensional partitions $\{\underline{\lambda}_\sigma\}_{\sigma \in \underline{\mathbf{3}}}$.
If there is no interaction between different directions organized by the $\mathscr{S}$ functions, it is just a union of three reverse plane partitions $\{ \mathsf{PP}_{\underline{\lambda}_\sigma} \}_{\sigma \in \underline{\mathbf{3}}}$ similar to $\mathrm{W}_{q_{1,2}}(A_1)$ algebra as seen in \S\ref{sec:double_screening_A1}.
Due to the $\mathscr{S}$ function factors, the pole structure turns out to be compatible with the PT3 count.

The conformal block \eqref{eq:block_F_multi} is identified with the partition function of gauge origami of D2 branes supported on $\{\mathbb{C}_{\sigma}\}_{\sigma \in \underline{\mathbf{3}}}$, whose matter content is given by the quiver shown in Fig.~\ref{fig:triple_F}.
For each $\sigma \in \underline{\mathbf{3}}$, we have the vertex function for the quasimaps to the Hilbert or the Quot scheme on the transversal directions $\bar{\sigma} = \underline{\mathbf{3}}\backslash\sigma$, and the total contribution is given by the equivariant Euler characteristics of the moduli space of the intersecting sheaves, that we call the origami vertex function.

\paragraph{MacMahon formula}

We finally consider the fully general situation.
The conformal block of the multi-screened vertex operator is given by the following contour integral,
\begin{align}
    \Psi_{\mathscr{C}_{1\ldots n}}^{\underline{k}}(\underline{a};\underline{K}) & = \bra{0} \mathsf{Z}^{(\underline{k}_1)}_{K_1}(a_1) \mathsf{Z}^{(\underline{k}_2)}_{K_2}(a_2) \cdots \mathsf{Z}^{(\underline{k}_n)}_{K_n}(a_n) \ket{0} / \bra{0} \mathsf{Z}_{K_1}(a_1) \mathsf{Z}_{K_2}(a_2) \cdots \mathsf{Z}_{K_n}(a_n) \ket{0}
    \nonumber \\ 
    & = \prod_{\sigma \in \underline{\mathbf{4}}} \frac{\mathsf{c}_{\sigma,\underline{k}}}{k_\sigma!} \oint_\gamma [\dd{\underline{x}}] \prod_{\sigma \in \underline{\mathbf{4}}} \Phi_{\sigma}^\text{M}(\underline{x}_\sigma,\underline{a}) \widetilde{\operatorname{Stab}}_{\sigma,\underline{k}}^\text{M}(\underline{x}_\sigma,\underline{a}) 
    \prod_{\sigma, \sigma' \in \underline{\mathbf{4}}} \prod_{\substack{i = 1,\ldots,k_\sigma \\ j = 1, \ldots, k_{\sigma'}}} \mathscr{S}_{\overline{\sigma\sigma'}}(q_\sigma x_{\sigma',j}/x_{\sigma,i})
    \, .
    \label{eq:block_M_multi}
\end{align}
For $\sigma \in \underline{\mathbf{4}}$ with $|q_\sigma| < 1$, the principal part and the stable envelope are understood through the analytic continuation $\left.\Phi_{\sigma}^\text{M}\right|_{q_{1,2,3,4} \to q_{1,2,3,4}^{-1}}$, $\widetilde{\operatorname{Stab}}_{\sigma,\underline{k}}^\text{M}|_{q_{1,2,3,4} \to q_{1,2,3,4}^{-1}}$.
Due to the condition \eqref{eq:CY4_cond}, we take for example $|q_{1,2,3}| > 1$ and $|q_4| < 1$.
In this case, the fourth-direction should be differently treated although the final result is expected to be symmetric for $q_{1,2,3,4}$.
Such an asymmetric treatment gives rise to a peculiar sign factor specific to the 4-fold counting~\cite{Nekrasov:2017cih,Nekrasov:2018xsb}.
The poles of the contour integral are parametrized by $n$-tuple four-dimensional partitions with a quadruple of boundary three-dimensional plane partitions $\{\underline{\pi}_\sigma\}_{\sigma \in \underline{\mathbf{4}}}$, and its structure is expected to be compatible with (the higher-rank version of) the four-leg PT4 count discussed in the literature~\cite{Cao:2019tvv,Monavari:2022rtf,Liu:2024bgp,Nekrasov:2023nai}.

The conformal block \eqref{eq:block_M_multi} is the partition function of gauge origami of D2 branes on $\{\mathbb{C}_\sigma\}_{\sigma \in \underline{\mathbf{4}}}$ described by the quiver shown in Fig.~\ref{fig:quadruple_M}.
Again from the geometric point of view, it would be identified with the fully general origami vertex function in $\mathbb{C}^4$.

\begin{conjecture}\label{conj:origami_vertex}
    The conformal blocks \eqref{eq:block_F_multi} and \eqref{eq:block_M_multi} are the rank-$n$ multi-leg PT vertex of $\mathbb{C}^3$~and~ $\mathbb{C}^4$.
\end{conjecture}

\appendix

\section{Special functions}


\begin{definition}[$q$-hypergeometric series, e.g.,~\cite{Gasper:2004}]\label{def:q-hypergeometric_fn}
    The $q$-hypergeometric series is defined as
    \begin{align}
        _r \phi_s
        \left[
        \begin{matrix}
            a_1 & \cdots & a_r \\ b_1 & \cdots & b_s
        \end{matrix}; q ; x
        \right]
        = \sum_{d \ge 0} \frac{(a_1,\ldots,a_r;q)_d}{(b_1,\ldots,b_s,q;q)_d} \left((-1)^d q^{{n \choose 2}}\right)^{1 + s - r} x^d 
        \, .
    \end{align}

\end{definition}
\begin{proposition}[Heine’s transformation formula]\label{prop:Heine_transform}
    Let $|x| < 1$.
    We have
    \begin{subequations}
    \begin{align}
    {_2\phi_1} \left[
    \begin{matrix}
        a & b \\ c & -
    \end{matrix}; q; x
    \right] 
    & = 
    \frac{(b,ax;q)_\infty}{(c,x;q)_\infty}
    {_2\phi_1} \left[
    \begin{matrix}
        c/b & x \\ a x & -
    \end{matrix}; q; b
    \right]
    \quad (|b| < 1) \\
    & = 
    \frac{(a,bx;q)_\infty}{(c,x;q)_\infty}
    {_2\phi_1} \left[
    \begin{matrix}
        c/a & x \\ b x & -
    \end{matrix}; q; a
    \right]
    \quad (|a| < 1) 
    \end{align}
    \end{subequations}
\end{proposition}

\begin{definition}[Elliptic gamma function]
For $p,q \in \mathbb{C}^\times$ with $|p|, |q| < 1$, we define
\begin{align}
    \Gamma(z;p,q) = \frac{(pqz^{-1};p,q)_\infty}{(z;p,q)_\infty} \, ,
    \label{eq:ell_gamma_fn}
\end{align}
obeying the functional relation
\begin{align}
    \frac{\Gamma(pz;p,q)}{\Gamma(z;p,q)} = \theta(z;q)
    \, , \qquad 
    \frac{\Gamma(qz;p,q)}{\Gamma(z;p,q)} = \theta(z;p)
    \, .
\end{align}    
\end{definition}
We write
\begin{align}
    \Gamma(a_1,\ldots,a_r;p,q) = \prod_{i=1}^r \Gamma(a_i;p,q) \, .
\end{align}


\begin{definition}[Elliptic hypergeometric series, e.g., \cite{Rosengren:2016qtr}]\label{def:E-hypergeometric}
    We define 
    \begin{align}
        {_{r+1} E_r} 
        \left[
        \begin{matrix}
            a_1 & \cdots & a_{r+1} \\ b_1 & \cdots & b_r
        \end{matrix}
        ; q, p; z
        \right]
        = \sum_{d \ge 0} z^d \frac{(a_1,\ldots,a_{r+1};q,p)_d}{(q,b_1,\ldots,b_{r};q,p)_d}
    \end{align}
    where
    \begin{align}
        (a_1,\ldots,a_{r};q,p)_d = \prod_{i=1}^r (a_i;q,p)_d
        \, , \qquad 
        (a;q,p)_d = \prod_{n=0}^{d-1} \theta(a q^n; p)
        \, .
    \end{align}
    The parameters should obey the balancing condition,
    \begin{align}
        a_1 a_2 \cdots a_{r+1} = q b_1 \cdots b_r \, .
        \label{eq:balance_cond}
    \end{align}
\end{definition}

\section{Vertex operator formalism}\label{sec:VO}

We summarize the vertex operator formalism used in the quantum algebraic approach to the gauge origami system~\cite{Kimura:2023bxy}.
See also~\cite{Noshita:2025bzg}.
We define the following sets, $\underline{\mathbf{4}} = \{1,2,3,4\}$, $\underline{\mathbf{6}} = \{12,13,14,23,24,34\}$, and $\underline{\mathbf{4}}^\vee = \{123,124,134,234\}$.

We define the vertex operators associated with D0, D2, D4, D6, and D8 branes as follows.
\begin{subequations}\label{eq:free_field}
\begin{align}
    \text{D0}: \quad & \mathsf{A}(x) = \mathsf{a}_0(x) :\exp\left(\sum_{n\neq0} \mathsf{a}_n x^{-n}\right): \, , \quad 
    [\mathsf{a}_n, \mathsf{a}_m] = - \frac{1}{n} (1-q_1^n)(1-q_2^n)(1-q_3^n)(1-q_4^n)\delta_{n+m,0} \, .
\end{align}
\begin{align}
    \text{D2}_\sigma: \quad & \mathsf{S}_\sigma(x) = \mathsf{s}_{\sigma,0}(x) :\exp\left(\sum_{n\neq0} \mathsf{s}_{\sigma,n} x^{-n}\right): \, , \hspace{2.7em} \mathsf{s}_{\sigma,n} = \frac{\mathsf{a}_n}{1-q_\sigma^{-n}} \, , \hspace{4.8em} \sigma \in \underline{\mathbf{4}} \, , \label{eq:D2_VO}\\    
    \text{D4}_A: \quad & \mathsf{X}_A(x) = \mathsf{x}_{A,0}(x) :\exp\left(\sum_{n\neq0} \mathsf{x}_{A,n} x^{-n}\right): \, , \qquad \mathsf{x}_{A,n} = \frac{\mathsf{a}_n}{\prod_{{\sigma} \in A}(1-q_{\sigma}^{-n})} \, , \quad A \in \underline{\mathbf{6}} \, , \\ 
    \text{D6}_{\bar{\sigma}}: \quad & \mathsf{W}_{\bar{\sigma}}(x) = \mathsf{w}_{\bar{\sigma},0}(x) :\exp\left(\sum_{n\neq0} \mathsf{w}_{\bar{\sigma},n} x^{-n}\right): \, , \hspace{1.6em} \mathsf{w}_{\bar{\sigma},n} = \frac{\mathsf{a}_n}{\prod_{\sigma \in \bar{\sigma}}(1-q_{\sigma}^{-n})} \, , \quad \bar{\sigma} \in \underline{\mathbf{4}}^\vee \, , \\  
    \text{D8}: \quad & \mathsf{Z}(x) = \mathsf{z}_0(x) :\exp\left(\sum_{n\neq0} \mathsf{z}_n x^{-n}\right): \, , \hspace{4.5em} \mathsf{z}_{n} = \frac{\mathsf{a}_n}{\prod_{\sigma \in \underline{\mathbf{4}}} (1-q_{\sigma}^{-n})} \, . \label{eq:D8_VO}
\end{align}
\end{subequations}    
See \cite[Appendix C]{Kimura:2023bxy} for the zero modes.
Then, the gauge origami partition function of $k$ D0 brane sector is given by the contour integral of the vacuum expectation value of the vertex operators,
\begin{align}
    \mathcal{Z}_k = \frac{\mathcal{G}^k}{k!} \oint [\dd{\underline{x}}] \langle 0 | \prod_{I=1}^k \mathsf{A}(x_I)^{-1} :\prod_{i} \mathsf{V}_i(v_i): | 0 \rangle \, , \qquad [\dd{\underline{x}}] = \prod_{I=1}^k \frac{\dd{x}_I}{2 \pi \ii x_I} \, ,
\end{align}
where $\mathsf{V}_i \in \{\mathsf{S}_\sigma,\mathsf{X}_A,\mathsf{W}_{\bar{\sigma}},\mathsf{Z}\}$ and $\mathcal{G}$ is a constant.
This contour integral is interpreted as the Jefferey--Kirwan residue associated with the rational factors appearing from the OPEs.
This formalism is generalized to the situation with non-trivial boundary conditions~\cite{Kimura:2024osv}.


\bibliographystyle{ytamsalpha}
\bibliography{references}

\end{document}